%% file: ex_article.tex
\documentclass[onefignum,onetabnum]{siamart171218}

\input{ex_shared}

\ifpdf
\hypersetup{
  pdftitle={Coded Computing for Fault-Tolerant Parallel QR Decomposition},
  pdfauthor={Quang Minh Nguyen, Iain Weissburg and   Haewon Jeong}
}
\fi

\myexternaldocument{ex_supplement}

\usepackage{amssymb}

\usepackage{graphicx}

\usepackage[algo2e, linesnumbered]{algorithm2e}

\usepackage{booktabs}
\usepackage{bbding}
\usepackage{pifont}
\usepackage{wasysym}

\input{import.tex}

\usepackage{enumitem}

\usepackage{lipsum}

\begin{document}

\maketitle

% REQUIRED
\begin{abstract}
  QR decomposition is an essential operation for solving linear equations and obtaining least-squares solutions. In high-performance computing systems, large-scale parallel QR decomposition often faces node faults. We address this issue by proposing a fault-tolerant algorithm that incorporates `coded computing' into the parallel Gram-Schmidt method, commonly used for QR decomposition.
Coded computing introduces error-correcting codes into computational processes to enhance resilience against intermediate failures. 
While traditional coding strategies cannot preserve the orthogonality of $Q$, recent work has proven a \emph{post-orthogonalization condition} that allows low-cost restoration of the degraded orthogonality. In this paper, we construct a checksum-generator matrix for multiple-node failures that satisfies the \emph{post-orthogonalization condition} and prove that our code satisfies the maximum-distance separable (MDS) property with high probability. 
Furthermore, we consider \emph{in-node checksum storage} setting where checksums are stored in original nodes. We obtain the minimal number of checksums required to be resilient to any $f$ failures under the in-node checksum storage, and also propose an in-node systematic MDS coding strategy that achieves the lower bound. Extensive experiments validate our theories and showcase the negligible overhead of our coded computing framework for fault-tolerant QR decomposition. 
\end{abstract}

% REQUIRED

% REQUIRED

\section{Introduction}
\input{intro}
\input{sysmod}

% \section{System Model and Problem Formulation}
% \input{problem_formulation}

\section{Checksum Generation for $R$-factor protection} 
\label{in_node}
\input{in_node}

\section{$Q$-Factor Protection}
\label{sec:q_factor_protection}
\input{q_factor_protection}

\section{Experimental Evaluation}
\input{experiment}

\section{Conclusion and Discussion}
\input{conclusion}

% \section*{Acknowledgments}
% We would like to acknowledge the assistance of volunteers in putting
% together this example manuscript and supplement.

\bibliographystyle{siamplain}
\bibliography{references}

\newpage

\appendix
% \appendix

\section*{APPENDIX}

\section{The PBMGS Algorithm and Its Overhead}
\label{pbmgs_appen}

\begin{algorithm}[h]
 \setcounter{AlgoLine}{0}
\lstset{language=Pascal} 
\lstset{numbers=left, numberstyle=\tiny, stepnumber=1, numbersep=5pt}
\caption{PBMGS($A, n, n, p_r, p_c$)}
\label{alg: PBMGS}
\SetAlgoLined
\KwIn{$n \times n$ matrix $A $ distributed 2D-block cyclically over $P=p_r \times p_c$ processors. }
\KwOut{$Q,R$}
\KwResult{$A = QR$, where $Q^T Q = I$}
Q=A 

\For{$i=1$ \KwTo $n$ \textbf{Step} $b$ }{ 
	$\bar{Q} = Q[:, i:i+b-1]$
 
	$ICGS(\bar{Q}, R[i:i+b-1, i:i+b-1])$
 
	$Q[:, i:i+b-1] = \bar{Q}$
 
   \If{$i < n-b$ }{ 
   
	All-reduce: $\bar{R} = \bar{Q}^T \cdot Q[:, i+b:n]$
 
	$R[i:i+b-1, i:i+b-1] = \bar{R}$
 
	$Q[:, i+b:n] -= \bar{Q} \cdot \bar{R}$
   }
}
\end{algorithm}

We consider the parallel blocked MGS algorithm (PBMGS) in   \cite{GS2}, which is specifically designed for 2D-block cyclic distribution. However, our analysis can be adapted to other algorithms  in \cite{givens_rotation2, GS1}.
At high level,  PBMGS (Algorithm \ref{alg: PBMGS}) iterates through $p_c = \frac{n}{b}$ block-columns. At the $i^{th}$ iterations it performs the iterated classical Gram-Schmidt algorithm (ICGS) on the whole block-column, and updates the remaining part (on the right side of the block-column) of the matrix. 
%As the update after ICGS involves all-reduce operation to perform matrix multiplication, it is the dominating cost. We use this cost as the lower bound for the PBMGS algorithm's cost. 
%If $T_{ICGS} (h, w, p)$ is the cost of the ICGS algorithm for QR-decomposing $h\times w$ matrix with $p$ vertical processors, then $T_{QR} \geq \frac{n}{b} T_{ICGS} (m, b, p_r)$

The ICGS algorithm is essentially the classical Gram-Schmidt (CGS), except that in each iteration ICGS further reorthogonalizes for one more time, depending on an easy-to-compute criterion \cite{criterion1}. The cost of ICGS algorithm is thus lower bounded by the cost of CGS algorithm (Algorithm \ref{alg: CGS}). To formally analyze the overhead of the PBMGS routine, we denote respectively by:

\begin{itemize}
    \item $T_{CGS}(A, n , b, p_r)$  the cost of performing CGS algorithm on the $n \times b$ matrix $A$ distributed 1D-block cyclically over $p_r$ processors,
\item and $T_{PMGS}(A, n , n, p_r, p_c)$  the cost of performing PMGSGS algorithm on the $n \times n$ matrix $A$ distributed 2D-block cyclically over $P=p_r \times p_c$ processors.
\end{itemize}

% If $T_{CGS} (h, w, p)$ is the cost of the CGS algorithm for QR-decomposing $h\times w$ matrix with $p$ vertical processors, then $T_{ICGS} (m, b, p_r) \geq T_{CGS} (m, b, p_r)$. 

%Therefore, $T_{QR} \geq \frac{n}{b} T_{CGS} (m, b, p_r) $. 

%We consider the sequential version of CGS in Algorithm \ref{alg: nCGS}, and further extend it to the parallel CGS in  Algorithm \ref{alg: CGS}. The cost  $ T_{CGS} (m, b, p_r)$ is formulated based on Algorithm \ref{alg: CGS}. 

In the $i^{th}$ iteration of algorithm $CGS(A, n , b, p_r)$, the computation of the $(b-i+1) \times 1$ vector $\bm{v}$ requires the local matrix-vector multiplication on each of the $p_r$ processors incurring $\gamma 2 \frac{n}{p_r} \cdot (b-i+1) $, and one MPI\_Reduce call with overhead $T_{reduce} (p_r, b-i+1)$. Summing up over $b$ iterations,  the cost of the CGS algorithm can be lower bounded by:
\begin{flalign*}
T_{CGS}(A, n , b, p_r)  & \geq \sum^{b}_{i=1}[ T_{reduce} (p_r, b-i+1) + \gamma 2 \frac{n}{p_r}  (b-i+1) ] \\
& = \sum^{b}_{i=1}[ 2\alpha \log p_r  +\beta \frac{p_r-1}{p_r} (b-i+1) \\
&+\gamma \frac{p_r-1}{p_r} (b-i+1)]+ \gamma \frac{nb(b+1)}{p_r}. 
\end{flalign*}
Hence,  
\begin{flalign}
T_{CGS}(A, n , b, p_r) &\geq 2\alpha b\log p_r   + \gamma \frac{nb(b+1)}{p_r}. \label{CGSbound1}
\end{flalign}

In the $i^{th}$ iteration of algorithm $PBMGS(A, n , n, p_r, p_c)$, the computation of the $b\times (n-i-b+1)$  matrix $\bar{R}$ requires the following main steps (we do not list all steps ) and communication patterns:
\begin{enumerate}
\item ICGS algorithm to update $\bar{Q}$, incurring $\geq T_{CGS}(A, n , b, p_r)$ overhead. 
\item Horizontal broadcast of $\bar{Q}$ via MPI\_Broadcast (in preparation for next step), incurring $T_{broadcast}(p_c, \frac{n}{p_r} b)$ overhead.
\item  Local matrix-matrix  multiplication:  each of the $p_r$ processors prepares its local inner-products through a
matrix-matrix multiplication of the local parts of $\bar{Q}$ (transposed) with its local parts of $Q$. This incurs $\gamma 2 b \frac{n}{p_r} \cdot \frac{n-i-b+1}{p_c} =\gamma 2 \frac{n}{P} b(n-i-b+1)$ overhead. 
\item MPI\_Reduce call with overhead $T_{reduce} (p_r, b(n-i-b+1))$ over $p_r$ vertical processors on each block-column. (Note: The local matrix-matrix multiplication in step 3 results in a $b\times (n-i-b+1)$ submatrix.)
\end{enumerate}

Summing up over all $\frac{n}{b}$ iterations (with step $b$), the cost of the PBMGS algorithm can be lower bounded by:
\begin{align*}
&T_{PBMGS}(A, n , n, p_r, p_c) \\
&\geq \frac{n}{b} T_{CGS}(A, n , b, p_r) + \frac{n}{b}  T_{broadcast}(p_c, \frac{n}{p_r} b) \\
&+ \sum^{\frac{n}{b}-1}_{\substack{j=0 \\ i =1+b j}}[T_{reduce} (p_r, b(n-i-b+1)) + \gamma 2 \frac{n}{P} b(n-i-b+1)]  \\
& \geq  \frac{n}{b}[2\alpha b\log p_r   + \gamma \frac{nb(b+1)}{p_r}] +\frac{n}{b} (\alpha \log p_c +\beta \frac{p_c-1}{p_c} \frac{n}{p_r} b) \\
&+ \sum^{\frac{n}{b}-1}_{\substack{j=0 \\ i =1+b j}}[\gamma 2 \frac{n}{P} b(n-i-b+1)]   \text{   (by (\ref{CGSbound1}))}\\
& \geq [2\alpha n \log p_r + \gamma \frac{n^2(b+1)}{p_r} ]+[\beta \frac{(p_c-1)n^2}{P}] \\
&+ \gamma (\frac{n^3}{P} -\frac{n^2b}{P} ) \\
&= 2\alpha n \log p_r +\beta \frac{(p_c-1)n^2}{P} + \gamma (\frac{n^3}{P} +\frac{n^2(b+1)p_c}{P}-\frac{n^2b}{P} ) \\
&\geq  2\alpha n \log p_r +\beta \frac{1}{2}\frac{p_cn(n+1)}{P} + \gamma \frac{n^2(n+1)}{P}.
\end{align*}
Therefore,
\begin{align}
\label{QR_bound1}
T_{QR} \geq  2\alpha n \log p_r +\beta \frac{1}{2}\frac{p_cn(n+1)}{P} + \gamma \frac{n^2(n+1)}{P}.
\end{align}

\begin{algorithm}[h]
 \setcounter{AlgoLine}{0}
\lstset{language=Pascal} 
\lstset{numbers=left, numberstyle=\tiny, stepnumber=1, numbersep=5pt}
\caption{CGS($A, n, b, p_r$)}
\label{alg: CGS}
\SetAlgoLined
\KwIn{$n \times b$ matrix $A $ distributed 1D-block cyclically over $p_r$ processors. }
\KwOut{$Q,R$}
\KwResult{$A = QR$, where $Q^T Q = I$}
Q=A 

\For{$i=1$ \KwTo $b$ }{ 
	$\bm{q} = Q[:, i]$ 
 
   All-reduce: vector	$\bm{v} = Q[:, i:b]^T \cdot Q[:,i]$
   
	$\bm{q}= \bm{q}/\sqrt{\bm{v}[1]}$
 
	$Q[:, i] = \bm{q}$ 
 
   	   	$R[i:, i:b] = \bm{v}^T$
        
   \If{$i < b$ }{ 
   	$Q[:, i+1:b] -= \bm{q} \cdot \bm{v}[i+1:b]^T$

   }
}
\end{algorithm}

\begin{observation}
\label{overhead2}
For the MGS algorithm \cite{GS2}, the overhead $T_{comp}$ for running QR decomposition on the larger $(m+c) \times n$ encoded matrix $\widetilde{A}$ instead of $m \times n$ matrix $A$ is bounded by:
\begin{align}
T_{comp} &\leq  \frac{c}{m} \cdot T_{QR}
\end{align}
\textit{Note:} We hereby derive the generalized result for $m \times n$ matrix A, and later set $m=n$ in our overhead analysis. 
\end{observation}
\begin{proof}[Proof of Observation \ref{overhead2}]
As any algorithm in \cite{givens_rotation2, GS1, GS2} iterates from left to right over the width of the matrix, there are $O(n)$ such iterations. In each iteration, all the mentioned algorithms use some form of MPI call (in Appendix \ref{MPI_oper}) excluding, MPI\_Alltoall. In particular, MPI\_Broadcast is required to communicate the data blocks, or MPI\_Reduce/MPI\_Allreduce is used to compute inner-products or matrix-products. As each processor owns $\frac{mn}{P}$ data, the data transferred in any such MPI call does not exceed $O(\frac{mn}{P})$. Any such MPI call (excluding MPI\_Alltoall) by Appendix \ref{MPI_oper} would incur $O(\alpha \log P )+ O(\beta \frac{mn}{P})$ communication. In terms of computation, all these efficient algorithms do not exceed $\gamma O(\frac{mn}{P})$ overhead. Therefore, summing up over $O(n)$ iterations, we can get the asymptotical cost for $T_{QR}$ as:
\begin{align}
\label{asym1}
T_{QR} = O(\alpha n \log P + \beta \frac{mn^2}{P} +\gamma \frac{mn^2}{P})
\end{align}

As $m$ increases at most linearly with the asymptotical cost of $T_{QR}$, the overhead $T_{comp}$ for running QR decomposition on some $(m+c) \times n$ matrix against some $m \times n$ matrix can be bounded by:
\begin{align*}
T_{comp} \leq \frac{c}{m} T_{QR}
\end{align*}

\begin{remark}
\label{no_latency1} 
In (\ref{asym1}), we also note that $T^{\alpha}_{comp} = O(1)$, as the increase of $m$ does not affect the $\alpha$ term. 
%Moreover, the asymptotical cost in (\ref{asym1}) gives $T^{\beta}_{comp} = O(\frac{cn^2}{P})$ and $T^{\gamma}_{comp} = O(\frac{cn^2}{P})$ .
\end{remark}
\end{proof}

\section{Checksum-preservation}
\label{appendixA:checksum_preserving}

%In this section, we discuss how horizontal and vertical checksums are respectively preserved as checksums for $R$-factor and $Q$-factor protection in MGS. We also consider BMGS, an alternative of MGS, as its parallel version PBMGS can better utilize the Level-3 operations in HPC \cite{criterion1} and thus be used in practice.  

We now proceed to prove the checksum-preservation, equations \eqref{eq:2way_checksum1} and \eqref{eq:2way_checksum2} in Lemma \ref{lem:2way_checksum}, for sequential MGS (Appendix \ref{appendixA:seqMGS}) and parallel MGS (Appendix \ref{appendixA:parMGS}) respectively. Then the proof of Corollary \ref{corol: final_form} is given in Appendix \ref{appen_corol_preservation1}.

\subsection{Checksum-preservation for sequential MGS}
\label{appendixA:seqMGS}

We first consider the columns of the encoded matrix $\widetilde{A}$:
\begin{align*}
&\widetilde{A} = \begin{bmatrix} \bm{\widetilde{a}_1} & \bm{\widetilde{a}_2} & \cdots & \bm{\widetilde{a}_{n+d}} \end{bmatrix} \text{, where $\bm{\widetilde{a}_i} = \begin{bmatrix} \bm{a_i} \\ G_v \bm{a_i}  \end{bmatrix} $} ,
\end{align*}
and the columns of the final $Q$-factor $Q^{(t)}$ at the $t^{th}$ iteration for $t=1 \rightarrow T$:
\begin{align*}
&Q^{(t)}= \begin{bmatrix} \bm{q_1}^{(t)} & \bm{q_2}^{(t)} & \cdots & \bm{q_{n+d}}^{(t)} \end{bmatrix},\\
&Q_1^{(t)}= \begin{bmatrix} \bm{q_{11}}^{(t)} & \bm{q_{12}}^{(t)} & \cdots & \bm{q_{1(n+d)}}^{(t)} \end{bmatrix},\\
&Q_2^{(t)}= \begin{bmatrix} \bm{q_{21}}^{(t)} & \bm{q_{22}}^{(t)} & \cdots & \bm{q_{2(n+d)}}^{(t)} \end{bmatrix},\\
&\text{where $\bm{q_i}^{(t)}=\begin{bmatrix} \bm{q_{1i}}^{(t)} \\ \bm{q_{2i}}^{(t)}  \end{bmatrix} $.}
\end{align*} 
For completeness, we now present the sequential MGS algorithm on $\widetilde{A}$ in Algorithm \ref{appendixA:alg: sequential_mgs}.
\begin{algorithm}
\setcounter{AlgoLine}{0}
\small
\lstset{language=Pascal} 
\lstset{numbers=left, numberstyle=\tiny, stepnumber=1, numbersep=5pt}
\caption{Sequential MGS}
\SetAlgoLined
\label{appendixA:alg: sequential_mgs}
\KwIn{$\widetilde{A} = \begin{bmatrix} \bm{\widetilde{a}_1} & \bm{\widetilde{a}_2} & \cdots & \bm{\widetilde{a}_{n+d}} \end{bmatrix}$ is $(n+c) \times (n+d)$ matrix}
\KwOut{$Q = \begin{bmatrix}\bm{q_1} & \bm{q_2} & \dots & \bm{q_{n+d}}\end{bmatrix}$ is $(n+c) \times (n+d)$  matrix  , and $R = (r_{ij})$ is $(n+d) \times (n+d)$  matrix}
\KwResult{$\widetilde{A} = QR$, where $Q^T Q = I$}
\For{$i=1$ \KwTo $n$ }{ 
    $\bm{u_i} = \bm{\widetilde{a}_i}$
    
    }
\For{$i=1$ \KwTo $n+d$ }{ 
   $r_{ii} = || \bm{u_i}||_2$ 
   
   $\bm{q_i} = \bm{u_i}/r_{ii} $
   
   \For{$j=i+1$ \KwTo $n+d$ }{ 
   
   	$r_{ij} = \bm{q_i}^T \bm{u_j}$ 
    
   	$\bm{u_{j}} = \bm{u_{j}} - r_{ij}\bm{q_i}$
   }
}
\end{algorithm}

In the Algorithm \ref{appendixA:alg: sequential_mgs}, we use the temporary $(n+c) \times (n+d)$ matrix $U = \begin{bmatrix} \bm{u_1} & \bm{u_2} & \cdots & \bm{u_{n+d}} \end{bmatrix}$. For analysis, we consider its value $U^{(t)}$ at the end of each iteration $t=1 \rightarrow T$. $U$ is initialized to $U^{(0)}= \widetilde{A}$ by the loop from line 1 to line 3. We further consider the breakdowns of the intermediate matrices as follows:
\begin{align*}
&U^{(t)}=\begin{bmatrix}U^{(t)}_1 \\ U^{(t)}_2 \end{bmatrix} 
= \begin{bmatrix} \bm{u_1} & \bm{u_2} & \cdots & \bm{u_{n+d}} \end{bmatrix}  \text{,  where $U^{(t)}_1: n \times (n+d)$, $U^{(t)}_2: c \times (n+d)$,}\\ &U^{(t)}_1=\begin{bmatrix} \bm{u_{11}}^{(t)} & \bm{u_{12}}^{(t)} & \cdots & \bm{u_{1(n+d)}}^{(t)} \end{bmatrix},\\
&U^{(t)}_2=\begin{bmatrix} \bm{u_{21}}^{(t)} & \bm{u_{22}}^{(t)} & \cdots & \bm{u_{2(n+d)}}^{(t)}
\end{bmatrix} \text{, where $\bm{u_i}^{(t)} = \begin{bmatrix}\bm{u_{1i}}^{(t)} \\ \bm{u_{2i}}^{(t)} \end{bmatrix}$}, \\
&R^{(t)} = \begin{bmatrix} R^{(t)}_1 & R^{(t)}_2\end{bmatrix}= (r^{(t)}_{ij}) = \begin{bmatrix} \bm{r_1}^{(t)} \\ \vdots \\ \bm{r_{n+c}}^{(t)} \end{bmatrix}, \\
&R^{(t)}_1= \begin{bmatrix} \bm{\bar{r}_{11}}^{(t)} \\ \vdots \\ \bm{\bar{r}_{(n+c)1}}^{(t)} \end{bmatrix} , 
R^{(t)}_2= \begin{bmatrix} \bm{\bar{r}_{12}}^{(t)} \\ \vdots \\ \bm{\bar{r}_{(n+c)2}}^{(t)} \end{bmatrix} \text{, where $\bm{r_i}^{(t)} = \begin{bmatrix} \bm{r_{i1}}^{(t)} & \bm{r_{i2}}^{(t)}\end{bmatrix}$}.
\end{align*}
For the sequential MGS (Algorithm \ref{appendixA:alg: sequential_mgs}), there are $T=n+d$ iterations in the main loop from line 4 to line 11. 
We proceed to prove by induction that for $t=0 \rightarrow T$, we have  $U^{(t)}_2 = G_v U^{(t)}_1$, $Q_2^{(t)}= G_v Q_1^{(t)}$ and $R_2^{(t)}=  R_1^{(t)} G_h$.

\textbf{Base case:} For $t=0$:\\
As $Q$ and $U$ are initialized to $Q^{(0)}= U^{(0)}=\widetilde{A}=\begin{bmatrix}
	 A & A G_h\\
	G_v A & G_v A G_h
\end{bmatrix}$, we have $Q_2^{(0)}= G_v Q_1^{(0)}$ and  $U_2^{(0)}= G_v U_1^{(0)}$. As $R$ is initialized to $R^{(0)} = 0$, we have $R_2^{(0)}=  R_1^{(0)} G_h=0$.

\textbf{Inductive step:} For the inductive step from $t-1$ to $t$ ($t\geq 1$): \\ 
Only the column $\bm{q_t}^{(t)}$ of $Q^{(t)}$ is updated on line 6:
\begin{align*}
    \bm{q_t}^{(t)} &= \bm{u_t}^{(t-1)}/r_{tt}=  \frac{1}{r_{tt}} \begin{bmatrix} \bm{u_{1t}}^{(t-1)} \\ \bm{u_{2t}}^{(t-1)}  \end{bmatrix} \\
    &=\frac{1}{r_{tt}} \begin{bmatrix} \bm{u_{1t}}^{(t-1)} \\ G_v \bm{u_{1t}}^{(t-1)}  \end{bmatrix} \text{ (as $U_2^{(t-1)}= G_v U_1^{(t-1)}$)}\\
    &=  \begin{bmatrix} \bm{u_{1t}}^{(t-1)}/r_{tt} \\ G_v ( \bm{u_{1t}}^{(t-1)}/r_{tt})  \end{bmatrix} = \begin{bmatrix} \bm{q_{1t}}^{(t)} \\ \bm{q_{2t}}^{(t)}  \end{bmatrix}
\end{align*}
We thus obtain: $\bm{q_{2t}}^{(t)} = G_v \bm{q_{1t}}^{(t)}$. Other columns of $Q$ remain the same: $\bm{q_{j}}^{(t)}= \bm{q_{j}}^{(t-1)}$, so $\bm{q_{2j}}^{(t)} = G \bm{q_{1j}}^{(t)}$ for all other $j\neq t$. We conclude that $Q_2^{(t)}= G_v Q_1^{(t)}$.\\
For the matrix $U$, there are $n-t$ columns $\bm{u_j}$ updated ( $j=t+1 \rightarrow n$) from line 7 to line 9: 
\begin{align*}
    \bm{u_j}^{(t)} &= \bm{u_j}^{(t-1)} - r_{tj} \bm{q_j}^{(t)}
    =\begin{bmatrix} \bm{u_{1j}}^{(t-1)} \\ \bm{u_{2j}}^{(t-1)}  \end{bmatrix} - r_{tj} \begin{bmatrix} \bm{q_{1j}}^{(t)}\\ \bm{q_{2j}}^{(t)}  \end{bmatrix}\\
    &=\begin{bmatrix} \bm{u_{1j}}^{(t-1)} \\ G_v \bm{u_{1j}}^{(t-1)} \end{bmatrix} - r_{tj} \begin{bmatrix} \bm{q_{1j}}^{(t)} \\ G_v \bm{q_{1j}}^{(t)} \end{bmatrix}\\
    &\text{ (as $U_2^{(t-1)}= G_v U_1^{(t-1)}$ and $Q_2^{(t)}= G_v Q_1^{(t)}$)}\\
    &= \begin{bmatrix} \bm{u_{1j}}^{(t-1)} - r_{tj}\bm{q_{1j}}^{(t)}   \\ G_v ( \bm{u_{1j}}^{(t-1)} - r_{tj}\bm{q_{1j}}^{(t)} ) \end{bmatrix} 
    = \begin{bmatrix} \bm{u_{1j}}^{(t)} \\ \bm{u_{2j}}^{(t)}  \end{bmatrix}
\end{align*}
We thus obtain: $\bm{u_{2t}}^{(t)} = G_v \bm{u_{1t}}^{(t)}$. Other columns of $U$ remain the same: $\bm{u_j}^{(t)}= \bm{u_j}^{(t-1)}$, so $\bm{u_{2j}}^{(t)} = G_v \bm{u_{1j}}^{(t)}$ for all other $j\leq t$. We conclude that $U_2^{(t)}= G_v U_1^{(t)}$.\\
Now, it is left to show that $R_2^{(t)}=  R_1^{(t)} G_h$. We note that only the row $\bm{r_t}^{(t)}$ of $R^{(t)}$ is updated on line 8, and have the following property.

\begin{observation}
\label{seqMGS_preserve}
At the end of iteration $t$, we have:
\begin{align}
\label{R_equation1}
\bm{r_t}^{(t)} &= \bm{q_t}^T U^{(t)} = \bm{q_t}^T \widetilde{A}
\end{align}
Note that here $\bm{q_t}$ refers to $\bm{q_t}^{(t)}$. However, as the value of $\bm{q_t}$ is last updated in the $t^{th}$ iteration and becomes the final value afterward, we neglect the superscript $(t)$ for simplicity.
\end{observation}
\begin{proof}
We first prove that $\bm{r_t}^{(t)} = \bm{q_t}^T U^{(t)}$. By line 5 and 8, $r^{(t)}_{tj} = \bm{q_t}^T \bm{u_j}^{(t)}$ for $j \geq t$, so it is left to show that $r^{(t)}_{tj} = \bm{q_t}^T \bm{u_j}^{(t)}$ also for $j < t$. We note that for $j < t$:
\begin{itemize}
    \item $r^{(t)}_{tj} = 0$ (as $r^{(t)}_{tj}$ is not updated by the algorithm and thus maintains the initial value $0$).
    \item $\bm{q_j}^{(j)} = \bm{u_j}^{(j)}/ r_{jj} $, updated on line 8 in the previous $j^{th}$ iteration.
    \item $\bm{q_j}^{(t)}= \bm{q_j}^{(j)}= \bm{q_j}$, as $\bm{q_j}^{(j)}$ is  last updated in the  $j^{th}$ iteration, and becomes the final value of $\bm{q_j}$. 
    \item $\bm{u_j}^{(t)}= \bm{u_j}^{(j)}= \bm{u_j}$, as $\bm{u_j}^{(j)}$ is  last updated in the  $j^{th}$ iteration, and becomes the final value of $\bm{u_j}$.
\end{itemize}
For $j < t$, we thus obtain that $\bm{q_t}^T  \bm{u_j}^{(t)}= r_{jj}\bm{q_t}^T\bm{q_j}$. As $Q = \begin{bmatrix}\bm{q_1} & \bm{q_2} & \dots & \bm{q_{n+d}}\end{bmatrix}$ is orthogonal, we have $I = Q^T Q = (\bm{q_i}^T\bm{q_j})_{ij}$, implying $\bm{q_t}^T\bm{q_j}=0$. Therefore, $\bm{q_t}^T \bm{u_j}^{(t)} = 0 = r^{(t)}_{tj}$. \\
To complete the proof of \eqref{R_equation1}, we proceed to prove that $\bm{q_t}^T U^{(t)} = \bm{q_t}^T \widetilde{A}$. This is equivalent to proving that $\bm{q_t}^T \bm{u_i}^{(t)} = \bm{q_t}^T \bm{\widetilde{a}_i}$ for any $i \in [1, n+d]$. Each $\bm{u_i}^{(t)}$ is only updated on line 9 in the first $i$ iterations, and thus formulated as:
\begin{align*}
    \bm{u_i}^{(t)} &= \bm{\widetilde{a}_i} - \sum_{l=1}^{\min(i,t)-1} r_{lj} \bm{q_l}^{(l)}\\
        &=  \bm{\widetilde{a}_i} - \sum_{l=1}^{\min(i,t)-1} r_{lj} \bm{q_l} \\
        &\text{ (as $\bm{q_l}^{(l)}$ is last updated in the $l^{th}$ iteration.)} 
\end{align*}
Left multiplying $\bm{q_t}^T$, we obtain:
\begin{align*}
    \bm{q_t}^T \bm{u_i}^{(t)}  &= \bm{q_t}^T \bm{\widetilde{a}_i} - \sum_{l=1}^{\min(i,t)-1} r_{lj} \bm{q_t}^T \bm{q_l} \\
    &= \bm{q_t}^T \bm{\widetilde{a}_i},
\end{align*}
where the last line holds because $I = Q^T Q = (\bm{q_i}^T\bm{q_j})_{ij}$, implying $\bm{q_t}^T\bm{q_j}=0$ for any $j\neq t$, and $l<\min(i,t) \leq t$.
\end{proof}

\noindent 
Back to our main proof, we now have:
\begin{align*}
\bm{r_t}^{(t)} &= \bm{q_t}^T \widetilde{A} = \bm{q_t}^T
\begin{bmatrix} A & A G_h \\ G_v A & G_v A G_h \end{bmatrix} \\
&=
\begin{bmatrix} \bm{q_t}^T\begin{bmatrix} A \\ G_v A \end{bmatrix} & \bm{q_t}^T\begin{bmatrix} A \\ G_v A \end{bmatrix} G_h  \end{bmatrix} \\
&=\begin{bmatrix} \bm{r_{t1}}^{(t)} & \bm{r_{t2}}^{(t)}\end{bmatrix}\\
\implies &\begin{cases}
\bm{r_{t1}}^{(t)} = \bm{q_t}^T\begin{bmatrix} A \\ G_v A \end{bmatrix}\\
\bm{r_{t2}}^{(t)} = \bm{q_t}^T\begin{bmatrix} A \\ G_v A \end{bmatrix} G_h
\end{cases} 
\end{align*}
We thus obtain: $\bm{r_{t2}}^{(t)} = \bm{r_{t1}}^{(t)} G_h$. Other rows of $R$ remain the same: $\bm{r_{j}}^{(t)}= \bm{r_{j}}^{(t-1)}$, so $\bm{r_{j2}}^{(t)} = \bm{r_{j1}}^{(t)} G_h$ for all other $j\neq t$. We conclude that $R_2^{(t)}=  R_1^{(t)} G_h$.

\subsection{Checksum-preservation for parallel MGS}
\label{appendixA:parMGS}

We consider the columns of the encoded matrix $\widetilde{A}$:
\begin{align*}
&\widetilde{A} = \begin{bmatrix} \bm{\widetilde{a}_1} & \bm{\widetilde{a}_2} & \cdots & \bm{\widetilde{a}_{n+d}} \end{bmatrix} \text{, where $\bm{\widetilde{a}_i} = \begin{bmatrix} \bm{a_i} \\ G_v \bm{a_i}  \end{bmatrix} $} 
\end{align*}
and the columns of the final $Q$-factor $Q^{(t)}$ at the $t^{th}$ iteration for $t=1 \rightarrow T$:
\begin{align*}
&Q^{(t)}= \begin{bmatrix} \bm{q_1}^{(t)} & \bm{q_2}^{(t)} & \cdots & \bm{q_{n+d}}^{(t)} \end{bmatrix}\\
&Q_1^{(t)}= \begin{bmatrix} \bm{q_{11}}^{(t)} & \bm{q_{12}}^{(t)} & \cdots & \bm{q_{1(n+d)}}^{(t)} \end{bmatrix}\\
&Q_2^{(t)}= \begin{bmatrix} \bm{q_{21}}^{(t)} & \bm{q_{22}}^{(t)} & \cdots & \bm{q_{2(n+d)}}^{(t)} \end{bmatrix} \text{, where $\bm{q_i}^{(t)}=\begin{bmatrix} \bm{q_{1i}}^{(t)} \\ \bm{q_{2i}}^{(t)}  \end{bmatrix} $.}
\end{align*} 
Next, we present the BMGS algorithm on $\widetilde{A}$:
\begin{algorithm}
\lstset{language=Pascal} 
\lstset{numbers=left, numberstyle=\tiny, stepnumber=1, numbersep=5pt}
\caption{BMGS}
\SetAlgoLined
\label{alg: bmgs}
\KwIn{$\widetilde{A} = \begin{bmatrix} \widetilde{a}_1 & \widetilde{a}_2 & \cdots & \widetilde{a}_n \end{bmatrix}$ is $(m+c) \times n$ matrix}
\KwOut{$Q = \begin{bmatrix}q_1 & q_2 & \dots & q_n\end{bmatrix}$ is $(m+c) \times n$  matrix  , and $R = (r_{ij})$ is $n \times n$  matrix}
\KwResult{$\widetilde{A} = QR$, where $Q^T Q = I$}
$Q=\widetilde{A}$\\
\For{$i=1$ \KwTo $n$ \textbf{Step} $b$}{ 
    $\bar{Q}=Q[:, i:i+b-1]$\\
    MGS($\bar{Q}, R[i:i+b-1, i:i+b-1]$)\\
    $Q[:, i:i+b-1]=\bar{Q}$\\
    \If{$i<n-b$}{
    $\bar{R}= \bar{Q}^T Q[:, i+b:n]$ \\
    $R[i:i+b-1, i+b:n]=\bar{R}$\\
    $Q[:,i+b:n] -= \bar{Q} \bar{R}$
    } 
}
\end{algorithm}

\noindent
The BMGS is essentially the reformulation of MGS beneficial for parallelization.
The parallel BMGS, called PBMGS, is obtained by parallelizing certain steps in BMGS via MPI operations to adapt to the 2D cyclic distribution, where step $b$ is set to the block size. As the computation and thus numerical output at each iteration of PBMGS remain the same as that of BMGS through the parallelization, it suffices to prove the checksum-preservation for BMGS. 
% Figure \ref{fig: bmgs} illustrates the checksum-preservation for PBMGS. 

\noindent
In the Algorithm \ref{alg: bmgs}, we use the temporary $(n+c) \times b$ matrix $\bar{Q}$. 
For analysis, we consider its value $\bar{Q}^{(t)}$ at the end of each iteration $t=1 \rightarrow T$. $\bar{Q}$ is initialized to $\bar{Q}^{(0)}= \widetilde{A}$ on line 1 of Algorithm \ref{alg: bmgs}.

\noindent
Next, we further consider the breakdowns of the
intermediate matrices as follows:

\begin{align*}
&\bar{Q}^{(t)}=\begin{bmatrix}\bar{Q}^{(t)}_1 \\ \bar{Q}^{(t)}_2 \end{bmatrix} \text{, where $\bar{Q}_1: n \times b$ and $\bar{Q}_2: c \times b$},\\
&R^{(t)} = \begin{bmatrix} R^{(t)}_1 & R^{(t)}_2\end{bmatrix}= (r^{(t)}_{ij}) = \begin{bmatrix} \bm{r_1}^{(t)} \\ \vdots \\ \bm{r_{n+c}}^{(t)} \end{bmatrix}, \\
&R^{(t)}_1= \begin{bmatrix} \bm{\bar{r}_{11}}^{(t)} \\ \vdots \\ \bm{\bar{r}_{(n+c)1}}^{(t)} \end{bmatrix} , 
R^{(t)}_2= \begin{bmatrix} \bm{\bar{r}_{12}}^{(t)} \\ \vdots \\ \bm{\bar{r}_{(n+c)2}}^{(t)} \end{bmatrix} \text{, where $\bm{r_i}^{(t)} = \begin{bmatrix} \bm{r_{i1}}^{(t)} & \bm{r_{i2}}^{(t)}\end{bmatrix}$}.
\end{align*}
We proceed to prove by induction that for $t=0 \rightarrow T$, we have $Q_2^{(t)}= G_v Q_1^{(t)}$ and $R_2^{(t)}=  R_1^{(t)} G_h$.

\textbf{Base case:} For $t=0$:\\
As $Q$ is initialized to $Q^{(0)}=\widetilde{A}=\begin{bmatrix}
	 A & A G_h\\
	G_v A & G_v A G_h
\end{bmatrix}$, we have $Q_2^{(0)}= G_v Q_1^{(0)}$.\\

\textbf{Inductive step:} For the inductive step from $t-1$ to $t$ ($t\geq 1$): \\
At the $t^{th}$ iteration, we have $i=1+tb$ in the for-loop of BMGS. 
The $Q$-factor is updated on line 5 for the $b$ columns $Q[:, i:i+b-1]$ and on line 9 for the $n-i-b+1$ columns $Q[:, i+b:n]$.
The matrix $\bar{Q}$ is initialized to $\bar{Q}= Q^{(t-1)}[:, i:i+b-1]$ on line 3 and then serves as the data input for the MGS algorithm. As $Q_2^{(t-1)}= G_v Q_1^{(t-1)}$, we obtain that $\bar{Q}_2= G_v \bar{Q}_1$ on line 3.
By the checksum preservation properties of MGS proven in Appendix \ref{appendixA:seqMGS}, the checksum relation $\bar{Q}_2= G_v \bar{Q}_1$ of $\bar{Q}$ is maintained throughout and at the end of the computation of MGS on line 4. Then, for the $b$ columns $Q[:, i:i+b-1]$ updated on line 5, we have:
\begin{align*}
    Q^{(t)}[:, i:i+b-1]&=\bar{Q}=
    \begin{bmatrix}\bar{Q}_1 \\ \bar{Q}_2\end{bmatrix}
    =\begin{bmatrix}\bar{Q}_1 \\ G_v \bar{Q}_1\end{bmatrix}\\
    &= \begin{bmatrix} Q^{(t)}_1[:, i:i+b-1] \\  Q^{(t)}_2[:, i:i+b-1]\end{bmatrix}.
\end{align*}
We thus obtain: $Q^{(t)}_2[:, i:i+b-1] = G_v Q^{(t)}_1[:, i:i+b-1]$.  
Now, for the $n-i-b+1$ columns $Q[:, i+b:n]$ updated on line 9:
\begin{align*}
    Q^{(t)}[:, i+b:n]&=Q^{(t-1)}[:, i+b:n]-\bar{Q} \bar{R}\\
    &= \begin{bmatrix} Q^{(t-1)}_1[:, i+b:n] \\  Q^{(t-1)}_2[:, i+b:n] \end{bmatrix} - \begin{bmatrix}\bar{Q}_1 \\  \bar{Q}_2 \end{bmatrix} \bar{R}\\
    &=\begin{bmatrix} Q^{(t-1)}_1[:, i+b:n] \\  G Q^{(t-1)}_1[:, i+b:n] \end{bmatrix} - \begin{bmatrix}\bar{Q}_1 \\ G \bar{Q}_1 \end{bmatrix} \bar{R}\\
    &\text{ (as $Q_2^{(t-1)}= G Q_1^{(t-1)}$ and $\bar{Q}_2= G \bar{Q}_1$)}\\
    &= \begin{bmatrix} Q^{(t-1)}_1[:, i+b:n] - \bar{Q}_1 \bar{R} \\ G ( Q^{(t-1)}_1[:, i+b:n] -\bar{Q}_1 \bar{R}) \end{bmatrix}\\
    &= \begin{bmatrix}Q^{(t)}_1[:, i+b:n] \\ Q^{(t)}_2[:, i+b:n] \end{bmatrix}. 
\end{align*}
We thus obtain: $Q^{(t)}_2[:, i+b:n] = G Q^{(t)}_1[:, i+b:n]$.\\
All in all, we have just proved that $q^{(t)}_{2j} = G q^{(t)}_{1j}$ for $j\in[i, n]$, where each of these columns are updated on either line 5 or line 9 of the BMGS algorithm. Other columns of $Q$ remain the same: $q^{(t)}_j= q^{(t-1)}_j$, so $q^{(t)}_{2j} = G q^{(t)}_{1j}$ for all other $j < i=1+tb$. Consequently, we conclude that $Q_2^{(t)}= G Q_1^{(t)}$. 
    
\noindent
Now, it is left to show that     $R_2^{(t)}=  R_1^{(t)} G_h$. We note that      the $R$-factor is updated on line 8 for the $b$ rows $R[i:i+b-1,:]$, and have the following property. 

\begin{observation}
\label{parMGS_preserve}
At the end of iteration $t$, we have:
\begin{align}
\label{eq:parMGS_preserve}
R^{(t)}[i:i+b-1, :] = \bar{Q}^T \widetilde{A}
\end{align}
\end{observation}
\begin{proof}
The same reasoning as Observation \ref{seqMGS_preserve}'s proof  can be adapted. Here we would like to present an intuitive proof to show the validity of \eqref{eq:parMGS_preserve}. 
From the algorithm, we know that $R^{(t)}[i:i+b-1, :]$, and $Q^{(t)}[:, i:i+b-1]=\bar{Q}$ are last updated in this $t^{th}$ iteration and thus also the final values:
\begin{align*}
    &R[i:i+b-1, :]=R^{(t)}[i:i+b-1, :], \\
    &Q[:, i:i+b-1]= Q^{(t)}[:, i:i+b-1]=\bar{Q}.
\end{align*}
If the algorithm correctly computes the QR decomposition of $\widetilde{A}$, at the end we would have $\widetilde{A}= QR$, implying:
\begin{align*}
    &R = Q^T \widetilde{A}, \\
    &R[i:i+b-1, :] = Q[:, i:i+b-1]^T  \widetilde{A} ,\\
    &R^{(t)}[i:i+b-1, :] = Q^{(t)}[:, i:i+b-1]^T  \widetilde{A} ,\\
    &R^{(t)}[i:i+b-1, :] = \bar{Q}^T  \widetilde{A} .
\end{align*}
\end{proof}

\noindent
Back to the main proof, we have: 
\begin{align*}
    R^{(t)}[i:i+b-1, :] &= \bar{Q}^T \widetilde{A} = \bar{Q}^T \begin{bmatrix} A & A G_h \\ G_v A & G_v A G_h \end{bmatrix} \\
&=
\begin{bmatrix} \bar{Q}^T\begin{bmatrix} A \\ G_v A \end{bmatrix} & \bar{Q}^T\begin{bmatrix} A \\ G_v A \end{bmatrix} G_h  \end{bmatrix} \\
&=\begin{bmatrix} R_1^{(t)}[i:i+b-1, :] & R_2^{(t)}[i:i+b-1, :] \end{bmatrix}\\
\implies &\begin{cases}
R_1^{(t)}[i:i+b-1, :] =  \bar{Q}^T\begin{bmatrix} A \\ G_v A \end{bmatrix}\\
R_2^{(t)}[i:i+b-1, :] =  \bar{Q}^T\begin{bmatrix} A \\ G_v A \end{bmatrix} G_h
\end{cases}. 
\end{align*}
We thus obtain: $\bm{r_{j2}}^{(t)} = \bm{r_{j1}}^{(t)} G_h$ for $j\in [i, i+b-1]$. Other rows of $R$ remain the same: $\bm{r_{j}}^{(t)}= \bm{r_{j}}^{(t-1)}$, so $\bm{r_{j2}}^{(t)} = \bm{r_{j1}}^{(t)} G_h$ for all other $j \notin [i, i+b-1]$. We conclude that $R_2^{(t)}=  R_1^{(t)} G_h$.
    
% \begin{figure}[t] 
% \centering
% \includegraphics[width=0.5\textwidth]{images/execution}
% \caption{$2^{nd}$ iteration of PBMGS, i.e. parallelized BMGS (Algorithm \ref{alg: bmgs}), for  $p_r=2, p_c =3$. Different colors represent different processors. (a) Six systematic processors (green, purple, red, blue, brown and white) own the original data blocks; three extra checksum-processors (orange, light orange, light green) own the vertical checksum-blocks $Q_{ij}$; and two extra  checksum-processors (red, pink) own the horizontal checksum-blocks $R_{ij}$
% (b) The left matrix illustrates  $Q$ , and the right matrix illustrates  $R$.} 
% \label{fig: bmgs}
% \end{figure}    
    
\subsection{Corollary \ref{corol: final_form}: final form of the QR decomposition on $\widetilde{A}$}
\label{appen_corol_preservation1}

By Lemma \ref{lem:2way_checksum}, we know that at the end of the algorithm, $Q_2^{(T)}=  G_v Q_1^{(T)}$ and $R_2^{(T)}=  R_1^{(T)} G_h$. 
As the QR decomposition is retrieved as $\widetilde{A}= QR$, where $Q= Q^{(T)} = \begin{bmatrix} Q_1^{(T)} \\ Q_2^{(T)} \end{bmatrix}= \begin{bmatrix} Q_1 \\ Q_2 \end{bmatrix}$ and $R= R^{(T)} = \begin{bmatrix} R_1^{(T)} & R_2^{(T)} \end{bmatrix}= \begin{bmatrix} R_1 & R_2 \end{bmatrix}$, equation \eqref{reduced_QR} directly holds: 
\begin{align*}
\widetilde{A} = \begin{bmatrix}.
Q_1 \\
Q_2
\end{bmatrix} \begin{bmatrix}
R_1 &
R_2
\end{bmatrix}
=
\begin{bmatrix}
Q_1 \\
G_v Q_1
\end{bmatrix} \begin{bmatrix}
R_1 &
R_1 G_h
\end{bmatrix}. 
\end{align*}

\section{Example for Code Construction \ref{const:G_formula1} and Theorem \ref{thm:prob_det}}
\label{construction_example_appen}

% We hereby present an example for illustration of  Code Construction \ref{const:G_formula1} and Theorem \ref{thm:prob_det}.

\begin{ex}[Code construction for $f=1$ and $n = p_r = p_c =2$]
Consider  the $2 \times 2$ matrix $A = \begin{bmatrix} A_{11} & A_{12} \\ A_{21} & A_{22}
\end{bmatrix}$, each entry of which is held by a  processor.
\end{ex}

\noindent
We now describe in details the steps of encoding, failure recovery and decoding.

\textbf{Encoding:} For $R$-factor protection, we use  $G_h = \begin{bmatrix} 1 & 1 \end{bmatrix}$, which is MDS. For $Q$-factor protection, we follow Construction \ref{const:G_formula1} to initiate $\widetilde{V} = [v]$ and $\widetilde{G}_1 = [ g_1]$ whereby $v \sim Unif(0, 1)$ and $g_1 = -\frac{1}{2} v^2$ (since $\widetilde{G}_1 = -\frac{1}{2} \widetilde{V} \widetilde{V}^T$). Given $\widetilde{G} = \begin{bmatrix} g_1 & v \end{bmatrix}$, we  construct the vertical checksum matrix in view of \eqref{generator_mat1} as $G_v = \widetilde{G} \otimes I_{n/p_r}= \begin{bmatrix} g_1 & v \end{bmatrix}$. Now, we encode the input matrix $A$ as:
{\small
\begin{equation} \label{eq:A_example}
   A 
    \xrightarrow{\text{Encode}} \widetilde{A} = \begin{bmatrix}
	 A & A G_h\\
	 G_v A & G_v A G_h
\end{bmatrix} = 
\begin{bmatrix} A_{11} & A_{12} & \widetilde{A}_{13} \\ A_{21} & A_{22} & \widetilde{A}_{23} \\ \widetilde{A}_{31} & \widetilde{A}_{32} & \widetilde{A}_{33}
\end{bmatrix}
\end{equation},
}\\
where, from our checksum design, the input matrix $A$ is protected via:
\begin{itemize}
    \item Horizontal checksums: $\widetilde{A}_{13} = A_{11} + A_{12}$ and $\widetilde{A}_{23} = A_{21} + A_{22}$ %and $A_{33} = A_{31} + A_{32}$
    \item Vertical checksums: $\widetilde{A}_{31} = g_1  A_{11} + v A_{21}$ and $\widetilde{A}_{32} = g_1  A_{12} + v A_{22}$ 
   % and $A_{33} = g_1  A_{13} + v A_{23}$
\end{itemize}

\textbf{Failure Recovery:} At the end of the first iteration of the MGS algorithm, we have the following view of the $Q$-factor and $R$-factor:
{\small
\begin{equation} 
   Q^{(1)} = 
\begin{bmatrix} Q_{11} & A_{12} & \widetilde{A}_{13} \\ Q_{21} & A_{22} & \widetilde{A}_{23} \\ Q_{31} & \widetilde{A}_{32} & \widetilde{A}_{33}
\end{bmatrix}, \quad R^{(1)} = 
\begin{bmatrix} R_{11} & R_{12} & R_{13} \\ 0 & 0 & 0 \\ 0 & 0 & 0
\end{bmatrix}
\end{equation}.
}\\
Now, from Lemma \ref{lem:2way_checksum}, we have $Q_{31} = g_1 Q_{11} + v Q_{12}$ and $R_{13} = R_{11} + R_{12}$, which, if complemented with the horizontal checksums of $A$ as in the encoding phase, would provide the full protection for $Q^{(1)}$ and $R^{(1)}$; e.g. assume that the processor $\Pi(1,1)$ holding $Q_{11}$ and $R_{11}$ fails. Then we can recover $Q_{11} = \frac{1}{g_1} (Q_{31} - v Q_{12})$ and $R_{11} = R_{13} - R_{12}$.

\textbf{Decoding:} At the end of the MGS algorithm,  we have the following view of the $Q$-factor and $R$-factor:
{\small
\begin{equation} 
   Q^{(3)} = 
\begin{bmatrix} Q_{11} & Q_{12} & Q_{13} \\ Q_{21} & Q_{22} & Q_{23} \\ Q_{31} & Q_{32} & Q_{33}
\end{bmatrix}, \quad R^{(3)} = 
\begin{bmatrix} R_{11} & R_{12} & R_{13} \\ 0 & R_{22} & R_{23} \\ 0 & 0 & 0
\end{bmatrix}
\end{equation}
}\\
from which we  retrieve $A = Q_1 R$ with $Q_1 = \begin{bmatrix} Q_{11} & Q_{12}  \\ Q_{21} & Q_{22} 
\end{bmatrix}$ and $R = \begin{bmatrix} R_{11} & R_{12}  \\ 0 & R_{22} 
\end{bmatrix} $. Next, we  construct  $G_0 = \begin{bmatrix} 1 + g_1 & v \\ v & -1\end{bmatrix}$  as in \eqref{G0_formula1}. Finally, to solve the full-rank square system $A \bm{x} = \bm{b}$, we compute the orthogonal matrix $G_0 Q_1$ and the transformed vector $G_0 \bm{b}$ in order to solve for $\bm{x}$ from \eqref{linearsolve3}.

\section{Omitted Proofs}
\input{app_in_node}

\end{document}

%% file: ex_shared.tex
\usepackage{lipsum}
\usepackage{amsfonts}
\usepackage{graphicx}
\usepackage{epstopdf}
\usepackage{algorithmic}
\ifpdf
  \DeclareGraphicsExtensions{.eps,.pdf,.png,.jpg}
\else
  \DeclareGraphicsExtensions{.eps}
\fi

\newsiamremark{remark}{Remark}
\newsiamremark{hypothesis}{Hypothesis}
\crefname{hypothesis}{Hypothesis}{Hypotheses}
\newsiamthm{claim}{Claim}

\headers{Coded QR Decomposition}{Quang Minh Nguyen, Iain Weissburg and   Haewon Jeong}

\title{Coded Computing for Fault-Tolerant Parallel QR Decomposition
\thanks{A preliminary version of this paper was presented at  IEEE International Symposium
on Information Theory 2020.}
}

\author{Quang Minh Nguyen \thanks{ Massachusetts Institute of Technology 
  (\email{nmquang@mit.edu}).}
\and Iain Weissburg  \thanks{ University of California, Santa Barbara
  (\email{ixw@ucsb.edu}).}
\and Haewon Jeong  \thanks{ University of California, Santa Barbara
  (\email{haewon@ucsb.edu}).}}

\usepackage{amsopn}

\makeatletter
\newcommand*{\addFileDependency}[1]{% argument=file name and extension
  \typeout{(#1)}% latexmk will find this if $recorder=0 (however, in that case, it will ignore #1 if it is a .aux or .pdf file etc and it exists! if it doesn't exist, it will appear in the list of dependents regardless)
  \@addtofilelist{#1}% if you want it to appear in \listfiles, not really necessary and latexmk doesn't use this
  \IfFileExists{#1}{}{\typeout{No file #1.}}% latexmk will find this message if #1 doesn't exist (yet)
}
\makeatother

\newcommand*{\myexternaldocument}[1]{%
    \externaldocument{#1}%
    \addFileDependency{#1.tex}%
    \addFileDependency{#1.aux}%
}

%% file: import.tex
\usepackage{hyperref} 
\usepackage{subcaption}
\usepackage{bm} 

\usepackage{comment}
\usepackage{graphicx}
\usepackage{amssymb}

\usepackage{mleftright}

\usepackage{fancyhdr}
\usepackage{mathtools}

\usepackage{listings}

\usepackage{algorithm, algorithmic}

\usepackage{tabularx}  
\usepackage{subcaption}

\input{style.tex}

\newcommand{\specialcell}[2][c]{%
  \begin{tabular}[#1]{@{}c@{}}#2\end{tabular}}

%% file: style.tex
\usepackage[noadjust]{cite}
\usepackage{color}
\usepackage{graphicx}
\usepackage{mathrsfs}
\usepackage{amsfonts}

\usepackage{amssymb}
\usepackage{bm}
\usepackage{listings}             % Include the listings-package

\usepackage{dsfont}

\newcommand{\ceil}[1]{\lceil #1 \rceil}
\newcommand{\floor}[1]{\lfloor #1 \rfloor}

\newtheorem{construction}{Construction}

\newtheorem{rem}{Remark}

\newtheorem{ex}{Example}

\newtheorem{observation}{Observation}

\theoremstyle{definition}

\def\mb0{\mathbf{0}}

\def\mbg{\mathbf{g}}

\def\mbv{\mathbf{v}}

\def \tilG{\widetilde{G}}
\def \tilg{\widetilde{g}}

\def \bbR{\mathbb{R}}

\newcommand{\Br}{\mathbb{R}}

%% file: intro.tex
Motivated by the widespread use of machine learning and scientific computing applications that require large-scale distributed computations, \emph{coded computing} has been a burgeoning area of research~\cite{lee2017speeding,dutta2016short, tandon2016gradient,yu2017polynomial,MatDotTransIT}. The aim of coded computing is to add systematically-designed redundancies in the computing algorithm to build fault resilience that is more efficient than replication or checkpointing. 
In this work, we consider protecting the QR decomposition algorithm from failures in high-performance computing (HPC) systems. Building a reliable supercomputer has been a long-standing problem, but the emergence of exascale computing poses new challenges that require a new and innovative solution that goes beyond traditional reliability techniques. More than 20\% of the computing capacity in today's HPC systems is wasted due to failures and ensuing recovery~\cite{Elnozahy:2008178}, and 
this wastage is only expected to grow as the system size grows. To reduce the overhead of fault tolerance in upcoming HPC systems, \emph{algorithm-based fault-tolerance (ABFT)} for HPC has been suggested \cite{BOSILCA2009410}, the core idea of which is essentially the same as coded computing: adding encoded redundancy tailored to a given numerical algorithm. 

QR decomposition factors a matrix into a product of an orthogonal matrix ($Q$) and an upper triangular matrix ($R$). It is an essential building block of linear algebraic computations as it provides a numerically stable method for solving linear equations and linear least squares problem (e.g. linear regression). 
As it is an important computation primitive, ABFT for parallel QR decomposition has been studied in the HPC literature~\cite{colchecksum, soft_error2, givens_rotations1}. While these existing works have shown that ABFT can be a more efficient resilience solution than periodic checkpointing, they proved that coding-based fault-tolerance can only be a half of the solution: while the $R$-factor can be protected with error-correcting codes (ECCs), the same does not apply to the $Q$-factor as encoding destroys the orthognal structure of the $Q$-factor~\cite[Theorem 5.1]{colchecksum}. Our prior work~\cite{qr_isit} introduced a novel strategy that can circumvent this fundamental theorem, by using a specialized code design that can restore the $Q$-factor's orthogonality with minimal overhead.However, this approach was limited to single-node failure scenarios, lacking a clear extension to multiple-node fault tolerance.

In this paper, we make critical contributions in generalizing the coded QR decomposition with low-cost post-orthogonalization to a wider variety of settings, as summarized below: 
\begin{itemize}
    \item 
    We propose a novel coding strategy for parallel QR decomposition that tolerates multiple-node failures with minimal overhead. Standard coding designs adequately protect the $R$-factor. However, protecting the $Q$-factor poses a challenge as its orthogonality is compromised by conventional linear coding~\cite{{colchecksum}}. We address this by devising a post-processing framework that restores $Q$'s orthogonality. Specifically, we derive a \emph{post-orthogonalization condition} which states that if a checksum-generator matrix follows  certain structure, there exists a sparse matrix $G_0$ that can linearly transform the coded $Q$ matrix into an orthogonal matrix. To this end, we design a semi-random checksum generator matrix satisfying the post-orthogonalization condition, while achieving the maximum-distance separable (MDS) property, i.e. code optimality, with high probability. The proposed algorithm can be applied to solving a full-rank square system of linear equations in HPC.
    
    \item We further consider the emerging setting of  \emph{in-node checksum storage} where  the coded data  is stored  in original systematic processors instead of being distributed over extra checksum processors for fault-tolerance. We prove a lower bound on the number of checksums required to be resilient to any $f$ failures and also demonstrate  coding schemes for both  $Q$-factor and $R$-factor protection that meet the  bound. The proposed in-node checksum coding strategy is  optimal, thereby improving over the results in \cite{colchecksum,qr_isit}. 
\end{itemize}

% study coded computing strategy for QR decomposition.

% However, in the previous literature, it was shown that we cannot protect the $Q$-factor by simply applying off-the-shelf error-correcting codes~\cite[Theorem 5.1]{colchecksum}. This is because linear encoding of the input matrix does not preserve the orthogonality of the $Q$-factor as a linear combination of orthogonal matrices is not guaranteed to be orthogonal. 
% In this paper, we propose a novel code design that allows us to recover the original $Q$-factor from the encoded output by performing a low-cost post-orthogonalization operation. In particular, we proved a \emph{post-orthogonalization condition} which states that if a checksum-generator matrix follows a certain structure, there exists a sparse matrix $G_0$ that can transform the coded $Q$ matrix into an orthogonal matrix. Furthermore, we show that a semi-random checksum generator matrix constructed to satisfy the post-orthogonalization condition has the MDS property with high probability. 

% A second contribution of the paper is improving the \emph{in-node checksum storage} result in \cite{colchecksum}. \emph{In-node checksum storage} setting means that checksums are stored in the original processors instead of in additional checksum processors. We prove a lower bound on the number of checksums required to be resilient to any $f$ failures and also demonstrate a coding scheme that meets the lower bound. The proposed in-node checksum coding strategy can be applied for $R$-factor protection. 

Finally, we note that this paper is an extension of our conference paper~\cite{qr_isit}. For both contributions, $Q$-factor protection codes and in-node checksum storage codes, we had a result for the single node failure case in \cite{qr_isit}. In this paper, we generalize those results to the multiple-node failure scenario. The results and constructions in \cite{qr_isit} can be regarded as a special case within the broader framework developed in this paper, tailored specifically for single-node failures.

%% file: sysmod.tex
\section{System Model and Problem Formulation}

We first describe our notations throughout the paper, and then present the system model and problem formulation within interests.

\textbf{Notations:}
We use upper-case letters to denote matrices and bold lower-case letters to denote vectors. For matrix $A$, we use $a_{i,j}$ to denote the $(i,j)$-th entry of the matrix. We use $I_m$ to denote the identity matrix of dimension $m \times m$, and write $I$ when the dimension is clear from the context.

% \noindent 
\textbf{Computation of interest:}
We consider solving a system of linear equations $A\bm{x}=\bm{b}$ through QR decomposition, where $A$ is a square and full-rank matrix of dimension $n \times n$. 
Solving square and non-singular system of linear equations is a fundamental building block for many applications in HPC~\cite{square_sys1, square_sys2, square_sys3}, and QR decomposition is a popular choice in practice due to its guaranteed  stability and computational efficiency~\cite{golub13}. 

% \noindent 
\textbf{Distributed computing model:} We assume that $A$ is a large matrix and the computation has to be distributed to multiple processors (nodes)\footnote{We use \emph{nodes} and \emph{processors} interchangeably in this paper.}. We assume that we have a total of $P$ nodes aligned as a $p_r \times p_c$ grid ($p_r p_c = P$).  We denote the $(i,j)$-th node on the grid as $\Pi(i,j)$. We consider \emph{2D block  distribution} of data where the matrix  $A$ is split into $p_r \times p_c$ blocks and each block $A_{ij}$ ($i\in [0, p_r-1], j\in [0,p_c-1]$) is distributed to processor $\Pi(i,j)$. The $j^{th}$ block-column is the "column" of data blocks $A_{0j}, A_{1j}, ..., A_{(p_r-1)j}$. Block-rows are defined similarly.  For simplicity, we assume that $p_r = p_c$ and consequently set the dimension of each block $A_{i,j}$ to be $b = \frac{n}{p_r} = \frac{n}{p_c}$. Note that $n \gg P$, and thus $n$ is much greater than $p_r$ or $p_c$. 

% \noindent 
\textbf{Failure model:} We assume that nodes are subject to failures, especially to \emph{``fail-stop errors.''} In this model, a failure corresponds to a node that completely stops responding, and loses all of its data, including its inputs and computation results. This is a realistic failure model in HPC~\cite{fail_stop1} and also similar to the \emph{``straggler model''} commonly used in the coded computing literature. 
At any iteration of the QR decomposition, we assume that at most $f$ failures can occur per column or per row, and that the identities of the processors that fail are provided by some external source (e.g. Message Passing Interface (MPI) library~\cite{MPI_detect_fault1}). Further, we assume that $f$ is  small  compared to $p_r, p_c$. Specifically, in our proposed scheme, we require $f \leq \frac{1}{2}\min\{p_r, p_c\}$.

% \noindent
\textbf{Checksums:}  In this work, we use checksums to recover from failures. Before we perform QR factorization, we encode the $n \times n$  matrix $A$ is as:
\begin{align}
\label{eq:encoding}
\widetilde{A} = \begin{bmatrix}
	 A & A G_h\\
	 G_v A & G_v A G_h
\end{bmatrix},
\end{align}
where $\widetilde{A}$ is of dimension $(n+c) \times (n+d)$ and $G_h$ and $G_v$ are checksum-generator matrices of size $n \times d$ and $c \times n$. We call $ G_v A$ \emph{``vertical checksums''} or \emph{``checksum rows''} and $ A G_h$
\emph{``horizontal checksums''} or \emph{``checksum columns''}. 
 After computation, vertical checksums protect $Q$ and horizontal checksums protect $R$~\cite[Lemma 1 and Corollary 1.1]{qr_isit}. 
We assume that $c$ and $d$ are small to keep the  overhead negligible.

Conventionally, to store the checksum part, we introduce additional nodes called \emph{``checksum nodes.''}  We call this setting \emph{out-of-node checksum storage}. 
In this setting, the original matrix $A$ is stored in \emph{``systematic nodes''} and all the checksum entries are stored in the \emph{checksum nodes.} 
Assuming that we use the same encoding coefficients at each node for the entire data block, $G_v$ and $G_h$ can be represented as: 
\begin{align}
    G_v  = \widetilde{G}_v \otimes I_{n/p_r}, \label{tildeG_eq} \;\; G_h  = \widetilde{G}_h \otimes I_{n/p_c},
\end{align}
where $\otimes$ is a Kronecker product, $\widetilde{G}_v$ and $\widetilde{G}_h$ are respectively $m_r \times p_r$  and $p_c \times m_c$ matrices, and  $m_r$ and $ m_c$ are respectively the number of checksum nodes per row and column.

In this paper, 
we will also consider an \emph{in-node checksum storage} setting, where one stores the checksums within the original set of nodes instead of introducing a separate set of checksum nodes. This will increase the size of the computing task per node, but it is better suited for a scenario where a user has a fixed number of available nodes.

\subsection{Problem Statement}
Our goal is to develop a coding scheme that can add effective checksums for parallel QR decomposition, which can protect both Q-factor and R-factor from any $f$ fail-stop errors ($f \geq 1$). This builds up on our previous work~\cite{qr_isit} that provided a coding scheme for a single fail-stop error ($f=1$) and generalizes to multiple errors and applies to both out-of-node and in-node checksum storage settings (see Table~\ref{table_of_summary}). We will provide  theoretical and experimental evidences to demonstrate the small overhead of the proposed coding scheme.

\begin{table}[t]
\small
    \begin{tabular}{c c c c c c c}
        \toprule
       \specialcell{Coding\\Scheme}   & QR Algorithm & $Q$-factor  & $R$-factor  & {$f > 1$} & Out-of-node & In-node\\
        \hline
       \cite{colchecksum }& Householder    & \ding{55} & \ding{51} & \ding{55} & \ding{55} &\ding{51}\\
       \cite{558063 }& Householder    & \ding{55} & \ding{51} & \ding{55} & \ding{51} &\ding{55}\\
       \cite{givens_rotations1} & Givens Rotation    & \ding{55} & \ding{51} & \ding{55} & \ding{51} &\ding{55}\\
       \cite{qr_isit} & Gram-Schmidt   & \ding{51} & \ding{51} & \ding{55}  & \ding{51} & \ding{51}\\
        This work & Gram-Schmidt   & \ding{51} & \ding{51} & \ding{51}  & \ding{51} & \ding{51}\\
        \bottomrule
        \end{tabular}
    \caption{Summary of work that utilize checksum-protection for QR decomposition. Those work \cite{colchecksum, 558063, givens_rotations1} without coding scheme for $Q$-factor protection rely on costly checkpointing for fault-tolerance.}
    \label{table_of_summary}
\end{table}

% The coding scheme is comprised of the encoding phase \eqref{R_checksum_preserving_equation}, the  recovery process that can tolerate at most $f$ fail-stop errors per column or row during the computation, and the decoding phase to retrieve the desired uncoded output. Finally, we provide theoretical guarantee on the robustness of our framework, and characterize its negligible overhead.

\section{Preliminaries and Related Work}
The core part of this work is incorporating error-correcting codes into a parallel QR decomposition algorithm for efficient checksum generation and failure recovery. For the QR factorization algorithm, we use parallel block modified Gram-Schmidt (PBMGS) algorithm. For error-correcting codes, we adapt the well-known maximum distance separable (MDS) codes. In this section, we describe essential preliminaries for these.

\begin{algorithm}[t]
% \small
% \begin{algorithmic}[1]
% \lstset{language=Pascal} 
% \lstset{numbers=left, numberstyle=\tiny, stepnumber=1, numbersep=5pt}
\caption{Modified Gram-Schmidt}
\SetAlgoLined
\label{alg: sequential_mgs}
\KwIn{$\widetilde{A} = \begin{bmatrix} \bm{\widetilde{a}_1} & \bm{\widetilde{a}_2} & \cdots & \bm{\widetilde{a}_{n+d}} \end{bmatrix}$ is $(n+c) \times (n+d)$ matrix}
\KwOut{$Q = \begin{bmatrix}\bm{q_1} & \bm{q_2} & \dots & \bm{q_{n+d}}\end{bmatrix}$ is $(n+c) \times (n+d)$  matrix  , and $R = (r_{ij})$ is $(n+d) \times (n+d)$  matrix}
\KwResult{$\widetilde{A} = QR$, where $Q^T Q = I$}
\For{$i=1$ \KwTo $n$ }{ 
    $\bm{u_i} = \bm{\widetilde{a}_i}$
    }
\For{$i=1$ \KwTo $n+d$ }{ 
   $r_{ii} = || \bm{u_i}||_2$ 
   
   $\bm{q_i} = \bm{u_i}/r_{ii} $
   
   \For{$j=i+1$ \KwTo $n+d$ }{ 
   	$r_{ij} = \bm{q_i}^T \bm{u_j}$ 
    
   	$\bm{u_{j}} = \bm{u_{j}} - r_{ij}\bm{q_i}$
   }
}
% \end{algorithmic}
\end{algorithm}

\subsection{Parallel block modified Gram-Schmidt (PBMGS) algorithm}
\label{sec:pbmgs}
There are three  classes of commonly used algorithms for QR decomposition: Gram-Schmidt (GS)~\cite{GS1, GS2, GS3}, Householder Transformation~\cite{blockhouseholder1, TSQR_Householder}, and Givens Rotation~\cite{givens_rotation2}. In this work, we consider parallel block modified Gram-Schmidt (PBMGS), which performs QR factorization on a grid of distributed nodes. For extreme-scale QR factorization, PBMGS is a preferred choice as it has a low computational cost and is easy to implement. 
% The algorithmic description of PBMGS and its overhead analysis are presented in Appendix \ref{pbmgs_appen}. 
In the Gram-Schmidt algorithm, we iteratively obtain a set of orthonormal vectors by subtracting the projections of each column onto the  space spanned by orthogonal vectors computed in the previous iterations. Modified Gram-Schmidt (MGS) performs this operation in a slightly different order to achieve better numerical stability (see Algorithm~\ref{alg: sequential_mgs}). Block MGS performs the same operation in column-blocks instead of column by column. How this can be parallelized in the PBMGS algorithm is depicted in Fig.~\ref{fig:pbmgs} (see Appendix~\ref{pbmgs_appen} for the full algorithm and its overhead analysis).

\begin{figure}[t]
    \centering
    \includegraphics[width=\textwidth]{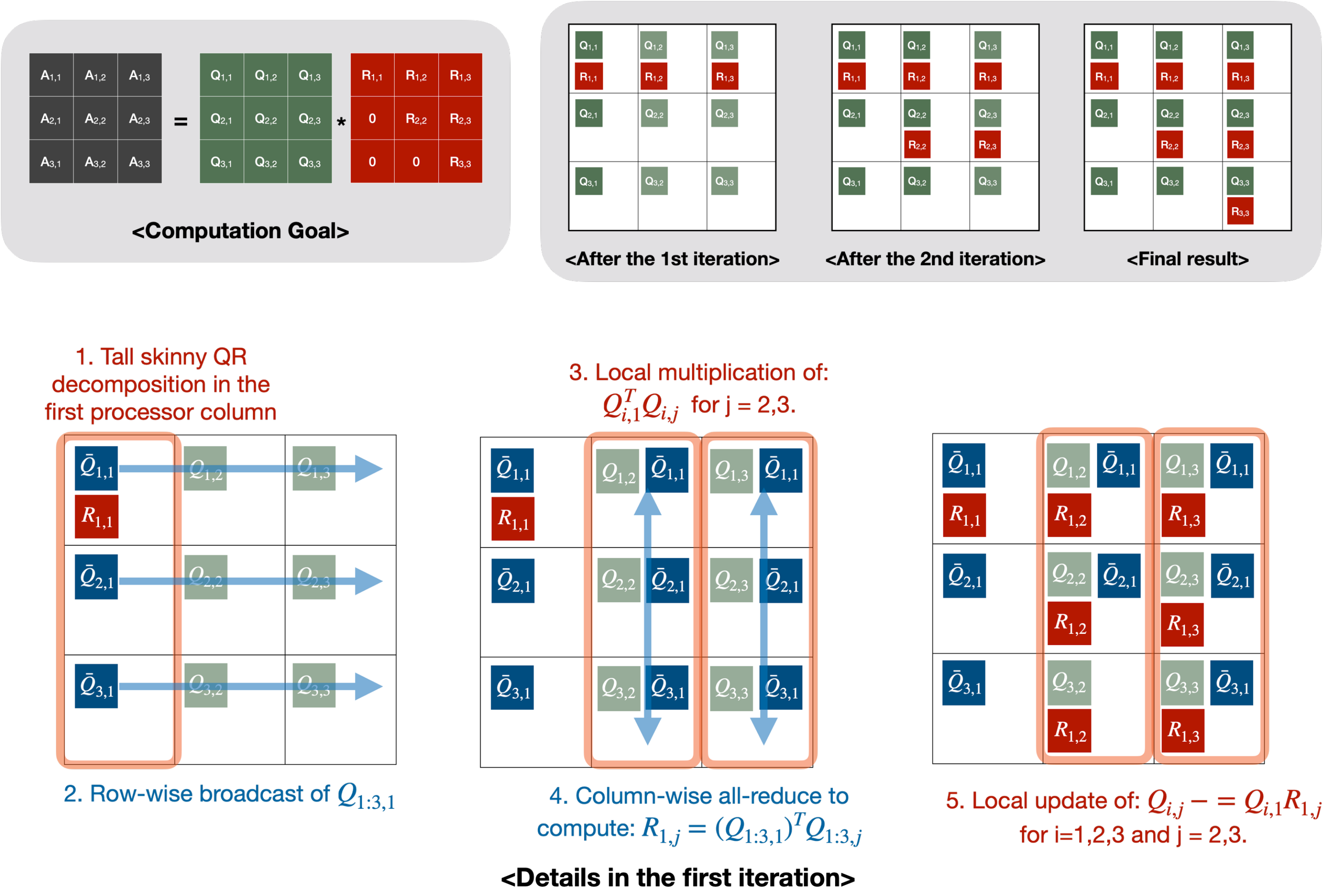}
    \caption{The parallel block modifed Gram-Schmidt (PBMGS) algorithm on a 3x3 grid.}
    \label{fig:pbmgs}
\end{figure}

\subsection{Maximum distance separable (MDS) codes}
What is the minimal number of checksum nodes we need to recover the data when $f$ nodes fail? This is a classic question in the coding theory literature. The information-theoretic optimum is $f$, given by the Singleton bound~\cite{joshi1958note}. MDS codes are the class of error-correcting codes that meet this bound~\cite{macwilliams1977theory}. MDS codes are linear block codes, meaning that each checksum is a linear combination of original entries and checksum generation can be represented with matrices as we did in \eqref{eq:encoding}. For codes to have the MDS property, checksum generator matrices should satisfy certain rank conditions. For our problem setup, these conditions can be translated as below: 
\begin{theorem}[Modified Theorem 8, p. 321 in \cite{macwilliams1977theory}]
$\widetilde{G}_v$ (resp. $\widetilde{G}_h$) has the MDS property if and only if every square submatrix of $\widetilde{G}_v$ (resp. $\widetilde{G}_h$) is full-rank.  
\end{theorem}

\subsection{Checksum preservation} \label{subsec:checksum_preserv}
We now review the desirable checksum preservation property of the PBMGS algorithm. When we perform QR decomposition on the encoded matrix $\widetilde{A}$, the PBMGS algorithm updates the $Q$ and $R$ factors in 
$T$ iterations. At the end of each iteration the algorithm maintains the intermediate $Q$-factor $Q^{(t)} $ of size $(n+c) \times (n+d)$ and $R$-factor $R^{(t)}$ of size $(n+d) \times (n+d)$. 
The initial values of the $Q$-factor and $R$-factor are set to be $Q^{(0)}=\widetilde{A}$ and $R^{(0)}=0_{(n+d) \times (n+d)}$. After the $T$-th iteration, we retrieve the orthogonal $Q=Q^{(T)}$ and the upper triangular $R=R^{(T)}$ as the output for the factorization $\widetilde{A} = QR$. 
We further consider the following submatrices of $Q^{(t)}$ and $ R^{(t)}$:
\begin{align*}
    &Q^{(t)} = \begin{bmatrix} Q_1^{(t)} \\ Q_2^{(t)} \end{bmatrix}, R^{(t)} = \begin{bmatrix} R_1^{(t)} & R_2^{(t)} \end{bmatrix}
\end{align*}
where $ Q_1^{(t)} $: $n \times (n+d)$,  $Q_2^{(t)} $: $c \times (n+d)$, $ R_1^{(t)} $: $(n+d) \times n$, and  $R_2^{(t)} $: $(n+d) \times d$. 
Lemma \ref{lem:2way_checksum} shows that the original vertical checksums $G_v A$ and horizontal checksums $A G_h$ are preserved after each iteration as $Q_2^{(t)}=G_v Q_1^{(t)}$ for $Q$-protection and $ R_2^{(t)}=R_1^{(t)} G_h$ for $R$-protection.
% The proofs for Lemma \ref{lem:2way_checksum} and Corollary \ref{corol: final_form} can be found in ~\cite[Appendix \ref{appendixA:checksum_preserving}]{full_version}. 

\begin{lemma}
\label{lem:2way_checksum}
For both sequential and parallel MGS, the following checksum relations hold for $t \in [0, T]$: 
\begin{align}
Q_2^{(t)}= G_v Q_1^{(t)}, \label{eq:2way_checksum1}\\
R_2^{(t)}=  R_1^{(t)} G_h. \label{eq:2way_checksum2}
\end{align}
\end{lemma}

\begin{corollary}
\label{corol: final_form}
At the end of the QR decomposition, the factorization of $\widetilde{A}$ has the final form:
\begin{align}
\label{reduced_QR}
\widetilde{A} = \begin{bmatrix}
Q_1 \\
Q_2
\end{bmatrix} \begin{bmatrix}
R_1 &
R_2
\end{bmatrix}
=
\begin{bmatrix}
Q_1 \\
G_v Q_1
\end{bmatrix} \begin{bmatrix}
R_1 &
R_1 G_h
\end{bmatrix}
\end{align}
where $Q_i=Q_i^{(T)}$ and  $R_i=R_i^{(T)}$.
\end{corollary}
The proofs of Lemma \ref{lem:2way_checksum} and Corollary \ref{corol: final_form} can be found in Appendix \ref{appendixA:checksum_preserving}.

\begin{comment}
\begin{proof}[Intuitive Proof:]
We present an intuitive proof for Equation (\ref{eq:2way_checksum}). From Algorithm \ref{alg: sequential_mgs}, we observe that each column $q_i$ of the $Q$-factor is on updated on line 6 by $q_i = u_i/r_{ii}$, so if the checksum relation holds for $u_i$, it also holds for $q_i$, i.e. for $q_i = \begin{bmatrix}q_{i1} \\ q_{i2}\end{bmatrix}$, we have $q_{i2} = G_v q_{i1}$. We can thus use induction on the iteration $t=0, \ldots, n+d$ to prove that the checksum relation holds for all $u_i$'s and $q_i$'s. The inductive step for $q_i$ is straight-forward. For $u_i$, it is only updated on line 9 by $u_j = u_j - r_{ij} q_i$. By the inductive hypothesis, we know that the checksum relation holds for $u_j$ in the previous  iteration and consequently $q_i$ in the current iteration, so we conclude the checksum relation holds for $u_j$ now. 
\end{proof}

\begin{corollary}[Corollary 1.1, \cite{qr_isit}]
\label{corol: final_form}
At the end of the QR decomposition, the factorization of $\widetilde{A}$ has the final form:
\begin{align}
\label{reduced_QR}
\widetilde{A} = \begin{bmatrix}
Q_1 \\
Q_2
\end{bmatrix} \begin{bmatrix}
R_1 &
R_2
\end{bmatrix}
=
\begin{bmatrix}
Q_1 \\
G_v Q_1
\end{bmatrix} \begin{bmatrix}
R_1 &
R_1 G_h
\end{bmatrix}
\end{align}
where $Q_i=Q_i^{(T)}$ and  $R_i=R_i^{(T)}$.
\end{corollary}
\end{comment} 

\begin{rem}
All the prevalent QR decomposition methods---Householder \cite{colchecksum, soft_error2}, Givens Rotations \cite{givens_rotations1}, and MGS---preserve the $R$-factor checksums. However, only MGS permits $Q$-factor checksum-preservation (i.e., satisfies \eqref{eq:2way_checksum1}) since it uses linear projection of the input matrix's column vectors for orthogonalization. On the other hand, Householder and Givens Rotations rely on non-linear operations such as reflections and rotations, which do not preserve the linearity of $Q$-factor checksums.
% Other prevalent QR decomposition methods including Householder \cite{colchecksum, soft_error2} and Givens Rotations \cite{givens_rotations1} do not possess $Q$-factor checksum-preservation, since their $Q$ matrix computation depends non-linearly on the column vectors of the input matrix $\widetilde{A}$. 

% While the R-factor checksum-preservation holds for Householder \cite{colchecksum, soft_error2}, Givens Rotations \cite{givens_rotations1} 

%     Remark after Lemma 1 explaining why checksum-preservation would not hold for other algorithms. 
\end{rem}

\subsection{Related work}
Algorithm-based fault tolerance (ABFT) has been proposed to alleviate the high overhead of checkpoint-restart (C/R) schemes in exascale high-performance applications~\cite{Chen:2006e6d, Bosilca:2008a57, Yao:20157d7, colchecksum, soft_error1, soft_error2, old_soft_error_1, old_soft_error_2, old_soft_error_3, LUchecksum_protected, givens_rotations1}. 
In the information theory community, a similar idea called \emph{coded computing} has been extensively studied, with more focus on renovating the code design under a theoretical computing model~\cite{yu2017polynomial,MatDotITTrans,dutta2016short,GC2,GC3,reisizadeh2017coded,jeong2018masterless,yu2017coded,Virtualization,ferdinand2016anytime,ferdinand2018hierarchical,mallick2018rateless,wang2018coded,quangQR,severinson2018block,NewsletterPaper}.
There exist several works on ABFT for QR decomposition~\cite{colchecksum, soft_error2,givens_rotations1, 558063}, and they are summarized in Table~\ref{table_of_summary}.
 For the out-of-node checksum storage setting, in \cite{558063, soft_error2}, the authors develop ABFT schemes for the Householder algorithm and in \cite{givens_rotations1}, they study ABFT for the Givens Rotation algorithm. The ABFT scheme for Householder is further adapted to the in-node checksum storage setting in \cite{colchecksum}. 
 However, all of them only protect the $R$ factor with checksums and rely on replication or checkpointing for the $Q$-factor. 
 In our prior work~\cite{qr_isit}, we proposed a coded MGS algorithm that can provide full protection for both $R$ and $Q$ factors when \emph{there is only one failure during computation ($f=1$)}. This work extends \cite{qr_isit} and provide a generalized scheme that can have resilience  for multiple failures.

\begin{figure}[t]
    \centering
    \includegraphics[width=\textwidth]{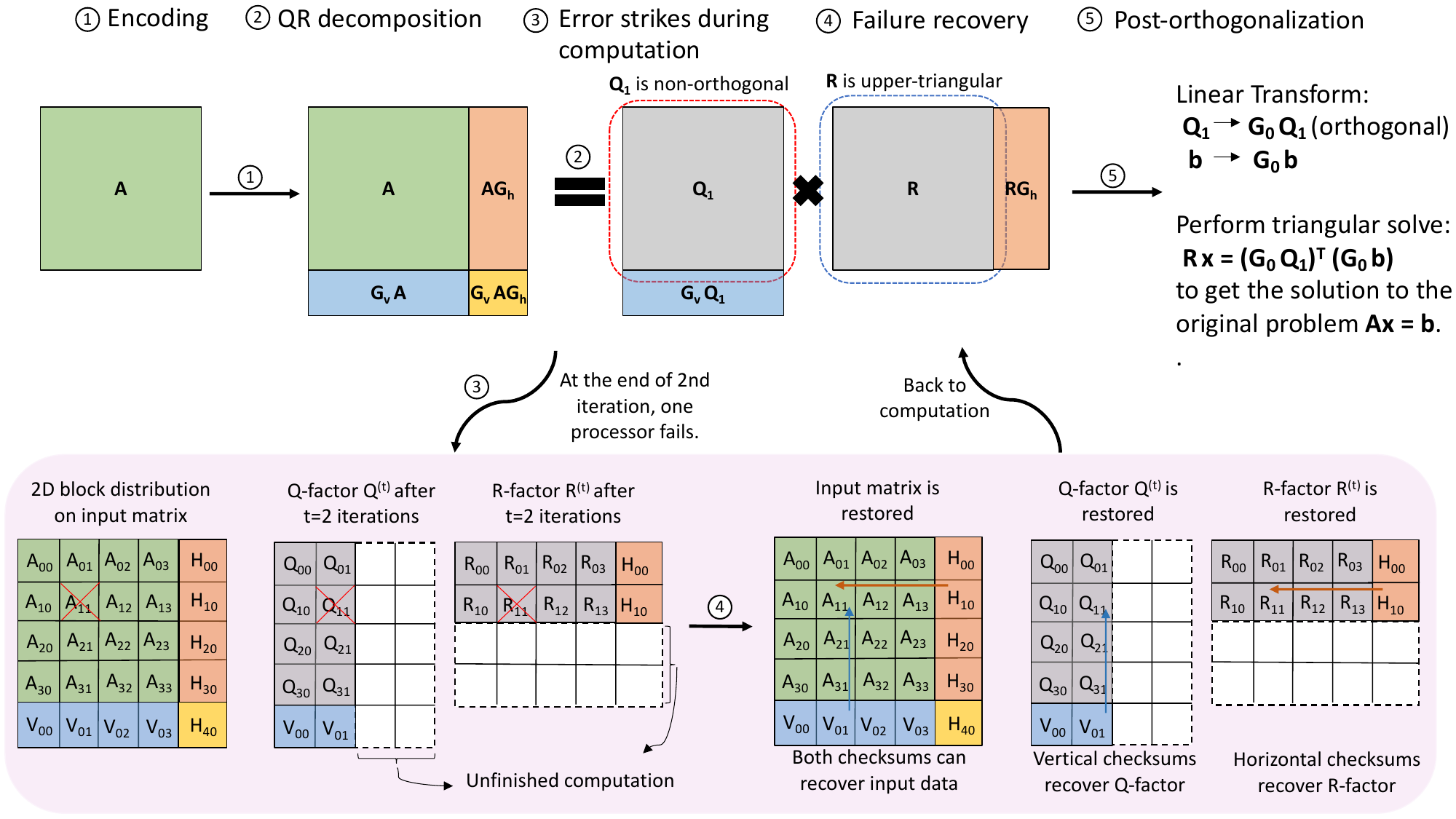}
    \caption{Coded QR decomposition framework.}
    \label{fig:qr_overview}
\end{figure}

\subsection{Overview of the Proposed Coded QR Decomposition Scheme}
We present the overall framework of coded QR decomposition in Figure \ref{fig:qr_overview}. It comprises of encoding, (coded) QR factorization, failure recovery (decoding), and post-orthogonalization steps.  
In the encoding phase, the input matrix $A$ is encoded  into $\widetilde{A}$ as given in \eqref{eq:encoding}. We then perform QR decomposition of the encoded matrix $\widetilde{A}$. If a node fails during the computation, we use vertical checksums to recover the intermediate $Q$-factor from the failed node and horizontal checksums to recover the $R$-factor. The original input matrix $A$ can be restored by either type of checksums. At the end of the computation, we retrieve $A = Q_1 R$. As we will show in Section~\ref{sec:q_factor_protection}, $Q_1$ is not orthogonal, and hence we deploy the post-orthogonalization technique which transforms $Q_1$ into $G_0 Q_1$ and $\mathbf{b}$ into $G_0 \mathbf{b}$, which are a very low-cost operations. After post-orthogonalization, we can solve the original problem: $A\mathbf{x} = \mathbf{b}$. 

% Central to

% \textcolor{red}{Perhaps table of comparision here?}

%% file: in_node.tex
% \textcolor{red}{Discussion on R-factor protection for out-of-node checksum storage and how people have done analysis.}

% \textcolor{red}{
In this section, we show how coded computing for $R$-factor protection can be extended from the out-of-node checksum storage setting to the in-node checksum setting. For the out-of-node setting, the coding process is straightforward. To protect from node failures for the $R$ factor, we can do row-wise encoding and add horizontal checksums $R G_h$ to checksum nodes as depicted in Fig.~\ref{fig:qr_overview}. Following Lemma~\ref{lem:2way_checksum}, any linear coding schemes can preserve the checksums of $R$ over the iterations in the PBMGS algorithm: $R_2^{(t)}=  R_1^{(t)} G_h$. Furthermore, $R_1^{(t)}$ still maintains the upper-triangular shape. Thus, there is no need for further processing if there is no failure. For optimal redundancy, we can use the checksum generator matrix of any MDS codes (e.g., Reed Solomon codes). We also refer to \cite{colchecksum,558063} for more details of out-of-node $R$ protection. 

The challenge we address in this paper is to extend the out-of-node checksum storage to the in-node checksum storage setting. 
In the coded computing literature, a common assumption is out-of-node checksum storage, which means that we add ``coded'' nodes that contain the encoded inputs. 
This follows the framework in coding theory where we have systematic symbols and parity-check symbols. However, in the HPC literature considers the in-node checksum storage setting for practical reasons. Many HPC algorithms are designed to perform optimally when the number of nodes is a power-of-two. If a fault-tolerant strategy requires changing the number of nodes, it can be a less appealing solution to practitioners. While practically appealing, in-node checksum storage imposes challenges on the design of checksums. Specifically, since the failure of a node now corresponds to the loss of not only input data, i.e. systematic symbols, but also checksums, i.e. parity-check symbols, the setting of which differs from traditional coding theory, thereby hindering direct adoption of existing MDS coding schemes. To this end, we dedicate this Section to derive a new  coding scheme for $R$-factor protection under the in-node checksum storage, and further show that it achieves the optimal amount of required redundancies. 
% }

In \cite{qr_isit}, we proved a lower bound on the number of checksum blocks and a scheme that achieves the lower bound for single failure recovery. This improved the existing strategy for in-node checksum storage ~\cite{colchecksum} by $\sim2$x. In this paper, we generalize this result for any number of failures $f \geq 1$. We first show a lower bound on the number of checksum blocks $K$ and then propose \emph{in-node systematic MDS coding} scheme that achieves the lower bound. 

Throughout this Section, we assume that nodes are load-balanced, and we make this assumption explicit for the checksums in the following.

% For the horizontal checksums $AG_h$, we can adopt any existing MDS checksum-generator matrix available in the literature since the resulting encoded matrix naturally preserves the upper triangular structure of the $R$ factor. On the other hand, if we apply off-the-shelf MDS encoding schemes for the vertical checksums $G_v A$, it does not preserve the orthogonality of the $Q$ factor~\cite[Theorem 5.1]{colchecksum}. The main challenge in this paper is thus designing a vertical checksum-generator matrix $G_v$ that can exhibit the MDS property and preserve orthogonality at the same time. 

% In this section, we show how coded computing strategies can be extended to the in-node checksum storage setting. 

% In the coded computing literature, a common assumption is out-of-node checksum storage, which means that we add ``coded'' nodes that contain the encoded inputs. This follows the framework in coding theory where we have systematic symbols and parity-check symbols. 

\begin{definition}[Load-balanced In-node Checksum Distribution]\label{def:innode_distr}
Let $P$ be the total number of nodes and let $K$ be the number of checksum blocks. Then, we say the checksum distribution is load-balanced if:
\begin{itemize}
    \item $\overline{a} \triangleq K - P\floor{\frac{K}{P}}$ of the nodes have $\ceil{\frac{K}{P}}$ checksum blocks, and
    \item $\overline{b} \triangleq P - K + P\floor{\frac{K}{P}}$ of the nodes have $\floor{\frac{K}{P}}$ checksum blocks.
\end{itemize}
\end{definition}

We now show a lower bound on $K$ in the next Theorem, whose proof can be found in Appendix \ref{appen_bound_K1}.

\begin{theorem} \label{thm:in_node_opt}
Let $L$ be the number of data blocks and assume that the total number of nodes $P$ divides $L$. Under the load-balanced in-node checksum distribution, to recover from any $r$ failed nodes out of $P$ nodes, 
\begin{equation} \label{eq:in_node_main}
K \geq \frac{fL}{P} + f \ceil{\frac{fL}{(P-f)P}}.    
 \end{equation}
\end{theorem}
In the following, we show that a simple adaptation of MDS coding with $K$ satisfying \eqref{eq:in_node_main} can tolerate any $f$ failures. 
\begin{construction}[In-Node Systematic MDS Coding] \label{const:in_node}
Let $D_1, \ldots, D_L$ be data blocks and $C_1, \ldots C_K$ be checksum blocks, and let $P$ be the number of nodes. We assume that $P$ divides $L$ and we assume the load-balanced setting, i.e, each node has $L/P$ data blocks and the checksum distribution follows the load-balanced in-node checksum distribution. Let $\mathcal{I}_p$ be the set of data block indices in the $p$-th node and let $\mathsf{Proc}(D_i)$  or $\mathsf{Proc}(C_i)$  denote the processor index the data or checksum block belongs to. 
Let $\tilG \in \bbR^{K \times L}$ be a checksum-generating matrix of a $(K+L, L)$ systematic MDS code. Then, In-Node Systematic MDS Coding encodes the $i$-th checksum block $C_i$ as follows:
\begin{equation} \label{eq:in_node_cksum}
    C_i = \sum_{j \notin \mathcal{I}_{\mathsf{Proc}(C_i)}} \tilG_{i,j} D_j.
\end{equation}
\end{construction}

% \textcolor{blue}{TODO: Update to include both R protection and Q protection codes!}

\begin{theorem} \label{thm:in_node_mds}
In-node systematic MDS coding given in Construction~\ref{const:in_node} with $K$ that satisfies \eqref{eq:in_node_main} can tolerate any $f$ failed nodes out of $P$ nodes. 
\end{theorem}

% Next, we apply the above generalized results for in-node checksum storage to the context of coded QR decomposition.
% In the following, we show that a simple adaptation of MDS coding with $K$ satisfying \eqref{eq:in_node_main} can tolerate any $f$ failures for both $Q$-factor and $R$-factor protection.

The proof of Theorem \ref{thm:in_node_mds} is given in Appendix \ref{appen_in_node_R1}.

% \textcolor{blue}{
\begin{rem}
\label{rem:alt_in_node}
We note that directly adopting the conventional MDS coding, i.e. encoding $C_i = \sum_{j=1}^L \tilG_{i,j} D_j$ instead of \eqref{eq:in_node_cksum},  also provides tolerance to $f$ failures following the similar argument as the proof of Theorem \ref{thm:in_node_mds}. Nevertheless, this would incur additional local computation, though negligible, for failure recovery due to the in-node structure. Particularly, we  note that $C_i - \sum_{j \in \mathcal{I}_{\mathsf{Proc}(C_i)}} \tilG_{i,j} D_j = \sum_{j \notin \mathcal{I}_{\mathsf{Proc}(C_i)}} \tilG_{i,j} D_j$. Thus, if $\mathsf{Proc}(C_i)$ is a surviving node, then it must locally compute the quantity $C_i - \sum_{j \in \mathcal{I}_{\mathsf{Proc}(C_i)}} \tilG_{i,j} D_j$, which is the numerical value of the weighted sum of other processors' data blocks, in the process of recovering failed nodes. To this end, our checksum construction \eqref{eq:in_node_cksum} alleviates such unnecessary step of local computation. 
\end{rem}
% }

% \begin{theorem} 
% % \label{thm:in_node_mds}
% The coding scheme given in Construction~\ref{const:in_node_qr} can tolerate any $f$ failed nodes in coded QR decomposition under in-node checksum storage. 
% \end{theorem}

%% file: q_factor_protection.tex
% For the horizontal checksums $AG_h$, we can adopt any existing MDS checksum-generator matrix available in the literature since the resulting encoded matrix naturally preserves the upper triangular structure of the $R$ factor. 
% \textcolor{red}{
We have demonstrated in Section \ref{in_node} that, as the resulting encoded matrix naturally preserves the upper triangular structure of the $R$ factor,  any existing MDS checksum-generator matrix and a well-designed MDS checksum-generator matrix would respectively suffice for the design of horizontal checksums $AG_h$ \eqref{eq:encoding} of out-of-node and in-node checksum storage. 
% }
On the other hand, if we apply off-the-shelf MDS encoding schemes for the vertical checksums $G_v A$, it does not preserve the orthogonality of the $Q$ factor~\cite[Theorem 5.1]{colchecksum}. The main challenge in this paper is thus designing a vertical checksum-generator matrix $G_v$ that can exhibit the MDS property and preserve orthogonality at the same time. 

In this section, we discuss how coding can be applied to the parallel Gram-Schmidt algorithm to protect the left $Q$ factor (an orthogonal matrix). From \eqref{eq:encoding}, we have both horizontal and vertical checksums, but in this section we will only consider the vertical checksum for simpler presentation. When we use $A$, one can regard it as $\begin{bmatrix}
A & AG_h
\end{bmatrix}$.

In \cite[Theorem 5.1]{colchecksum}, it was concluded:  \emph{``$Q$ in Householder QR factorization cannot be protected by performing factorization along with the vertical checksum.'' } The basis of this claim was that in the output retrieval $A = Q_1 R$ after the QR factorization of the vertically encoded matrix $\widetilde{A}$:
\begin{align}
\label{encoded_qr_compute}
    \widetilde{A} =\begin{bmatrix}
    A \\ GA
    \end{bmatrix} =\begin{bmatrix}
    Q_1 \\ Q_2
    \end{bmatrix} R,
\end{align}
$Q_1$ is not orthogonal, i.e. $Q_1 ^T Q_1 \neq I$. %This is because linearly combining rows do not preserve the orthogonality. 
Thus, $Q_1 R$ is not  the correct QR factorization of $A$. While the Theorem statement was proven for the Householder algorithm, the reasoning is general and not limited to a specific algorithm. 
% \textcolor{red}{[TODO] Refine the sentence and add the theorem? }

A major contribution of our work is showing that we can convert $Q_1$ into an orthogonal matrix with a small amount of overhead (same as the overhead of encoding in scaling sense). 
In Section~\ref{subsec:general_scheme}, we prove that if the checksum-generator matrix satisfies certain conditions (given in Theorem~\ref{thm:G_0}), there exists a low-cost linear transform that orthogonalizes $Q_1$. In Section~\ref{multiple_failure}, we propose a checksum-generator matrix construction for out-of-node checksum storage that satisfies the conditions given in Theorem~\ref{thm:G_0} while providing resilience to multiple-node failure. 
Then, in Section \ref{q_in_node}, we adapt the checksum-generator matrix construction to the setting of in-node checksum storage.
% Finally, we show through careful analysis that the overhead of fault tolerance including encoding, failure recovery, and low-cost post-orthogonalization is negligible. 

\subsection{Low-cost post-orthogonalization}
\label{subsec:general_scheme}

How ``non-orthogonal'' is $Q_1$? Can we still utilize $Q_1$ to recover the original QR factorization? We show that with a low-cost linear transform, $Q_1$ can be transformed into an orthogonal matrix, if the checksum-generator matrix $G$ satisfies certain conditions. 

\begin{theorem}\label{thm:G_0}
Let $G_1, V$ be submatrices of the vertical checksum-generator matrix $G$ as follows:
\begin{equation}
    G = \begin{bmatrix}  G_1 & V \end{bmatrix},
\end{equation}
where $G_1$ and $V$ have dimensions $c \times c$ and $c \times (n-c)$, respectively. Let $G_0$ be an $n \times n$ by matrix as follows:
\begin{equation}
\label{G0_formula1}
    G_0 = \begin{bmatrix}
I_{c}+G_1 & V \\
V^T & -I_{n-c} \end{bmatrix}.
\end{equation}

If $G$ satisfies the following condition:
\begin{align}
\label{eq:restrictionG}
G_1 = -\frac{1}{2} V V^T,
\end{align}
we can prove the following:

Claim 1: $G_0 Q_1$ is orthogonal, i.e.  $(G_0 Q_1)^T (G_0 Q_1) = I$.

Claim 2:  $G_0$ is invertible.

% \begin{enumerate}
%     \item $G_0 Q_1$ is orthogonal, i.e.  $(G_0 Q_1)^T (G_0 Q_1) = I$.
%     \item $G_0$ is invertible.
% \end{enumerate}
\end{theorem}

% The proof for Theorem \ref{thm:G_0} is in ~\cite[Appendix~\ref{Q_lemma}]{full_version}. \todo{note}

The proof of Theorem \ref{thm:G_0} is given in Appendix \ref{appen_post_ortho_proof1}. 
 Based on Theorem \ref{thm:G_0}, we next show how a checksum generator matrix satisfying \eqref{eq:restrictionG} can be used to retrieve the orthogonality of the original $Q$-factor. In Section \ref{multiple_failure},  we will  design a semi-random checksum generator matrix (Construction~\ref{const:G_formula1}) that jointly satisfies \eqref{eq:restrictionG} and is MDS with high probability (Theorem \ref{thm:prob_det}). 
Now, notice that the matrix $G_0$ is very sparse as the bottom-right submatrix is simply an $(n-c)\times (n-c)$ identity matrix, and $c$ is negligible. Claim 1 in Theorem~\ref{thm:G_0} suggests that by multiplying this sparse matrix $G_0$, we can convert $Q_1$ into an orthogonal matrix. 
We now demonstrate how we can use $G_0 Q_1$ in place of the original $Q$ in solving a full-rank square system of linear equations $A\bm{x} = \bm{b}$.
Let $A' = G_0 A$. Then the following factorization is a valid QR factorization of $A'$:
\begin{equation} \label{eq:A'}
    A'  = (G_0 Q_1) R,
\end{equation}
with the left factor $(G_0 Q_1)$ and the right factor $R$. As $G_0$ is invertible by Claim 2 in Theorem~\ref{thm:G_0}, 
\begin{equation}
    A\bm{x} =\bm{b} \iff (G_0 A)\bm{x} = G_0 \bm{b}.
\end{equation}
The linear system on the right side can be solved using the QR factorization of $A'$ given in \eqref{eq:A'}. i.e.,
\begin{align}
    (G_0 Q_1) R\bm{x} & = G_0 \bm{b},  \label{linearsolve1}\\ 
   (G_0 Q_1)^T (G_0 Q_1) R\bm{x} &= (G_0 Q_1)^T (G_0 \bm{b}), \label{linearsolve2}\\
     R\bm{x} &= (G_0 Q_1)^T (G_0 \bm{b}). \label{linearsolve3}
\end{align}
Then, we can perform triangular solve to get the final answer $\bm{x}$. Remember that we already have $Q_1$ and $R$ from the QR factorization of the encoded matrix $\widetilde{A}$. Hence, all we need to perform in the above steps is computing post-orthogonalization, $G_0 Q_1$ and $G_0 \bm{b}$. We show in Theorem~\ref{thm:overhead} that the overhead of post-orthogonalization is negligible. 

% \hj{Repetition of theorem. Replication is not the most effective way of redundancy :P}
% \begin{theorem}
% $\tilde{G}$ has the MDS property if and only if every square submatrix of $\tilde{G}$ is full-rank. 
% \end{theorem}

% \subsection{Checksum-Generator Matrix for Multiple-Node Failure} 
\subsection{Checksum-Generator Matrix for Out-of-Node Checksum Storage} 
\label{multiple_failure}

The low-cost post-orthogonalization scheme exists under the constraint \eqref{eq:restrictionG} on the checksum-generator matrix $G$. One crucial question  is whether we can construct  $G$  that has good error correction/detection capability while satisfying \eqref{eq:restrictionG}. In this subsection, we present one such construction of $G$ for $f$ node failures recovery. 
Throughout this section, we assume  $p_r$ divides $n$ for simplicity, but results generalize to any $p_r$ and $n$. %\hj{How does it generalize?} 

Recall from  \eqref{tildeG_eq} that we can write $G = \widetilde{G} \otimes I_{\frac{n}{p_r}}$. Let $\tilde{G_1}, \tilde{V}$ be submatrices of $\widetilde{G}$ such that $\widetilde{G} = \begin{bmatrix}  \widetilde{G}_1 & \widetilde{V} \end{bmatrix}  $, where $\widetilde{G}_1$ and $\widetilde{V}$ have dimensions $m_r \times m_r$ and $m_r \times (p_r-m_r)$, respectively. Since ${G} = \begin{bmatrix}  {G}_1 & {V} \end{bmatrix}  $ where ${G}_1$ and ${V}$ have dimensions $c \times c$ and $c \times (n-c)$, we have: $G_1 = \widetilde{G}_1 \otimes I_{n/p_r}$ and $V = \widetilde{V} \otimes I_{n/p_r}$. The following Lemma establishes the equivalence of the post-orthogonalization condition being imposed on $G$ versus $\widetilde{G}$.
\begin{lemma}[Equivalence of post-orthogonalization condition]
The following relation holds:
\begin{align}
    G_1 = -\frac{1}{2} V V^T \iff \widetilde{G_1} = -\frac{1}{2} \widetilde{V} \widetilde{V}^T.
\end{align}

\end{lemma}
\begin{proof}
We have:
\begin{align*}
     G_1 = -\frac{1}{2} V V^T  &\iff  \widetilde{G_1} \otimes I_{n/p_r} = -\frac{1}{2} (\widetilde{V} \otimes I_{n/p_r}) (\widetilde{V}^T \otimes I_{n/p_r}) \\
     & \iff  \widetilde{G_1} \otimes I_{n/p_r}  = (-\frac{1}{2} \widetilde{V} \widetilde{V}^T) \otimes I_{n/p_r}  \\
     & \iff  \widetilde{G_1} = -\frac{1}{2} \widetilde{V} \widetilde{V}^T.
\end{align*}
\end{proof}

\begin{construction}[$Q$-factor Checksum-Generator Matrix for Multiple-Node Failure Recovery]
\label{const:G_formula1}

Let $\widetilde{G} = \begin{bmatrix} \widetilde{G}_1 & \widetilde{V} \end{bmatrix} = \{ \tilg_{i,j}\}$ be a $f \times p_r$ matrix, where $\widetilde{V}$ is a $f \times (p_r-f)$ matrix whose entries are chosen randomly from $\mathrm{Unif}(0,1)$ and $\widetilde{G}_1 = -\frac{1}{2} \widetilde{V} \widetilde{V}^T$ is an $f\times f$ matrix. 

The following $\frac{f n}{p_r} \times n$ checksum-generator matrix $G$ satisfies the restriction \eqref{eq:restrictionG}, and, as established in Theorem \ref{thm:prob_det}, guarantees multiple-node fault tolerance for out-of-node checksum  storage with probability 1: 
\begin{align}
\label{generator_mat1}
G =\widetilde{G} \otimes I_{n/p_r} = \begin{bmatrix} \tilg_{0,0} I_{\frac{n}{p_r}} & \tilg_{0,1}I_{\frac{n}{p_r}}  & \cdots & \tilg_{0,p_r-1} I_{\frac{n}{p_r}} \\ 
\vdots  &\ddots & & \vdots \\
\tilg_{f-1,0} I_{\frac{n}{p_r}} & \tilg_{f-1,1}I_{\frac{n}{p_r}}  & \cdots & \tilg_{f-1,p_r-1} I_{\frac{n}{p_r}} \end{bmatrix}.
\end{align}

\end{construction}

\newcommand\coolover[2]{\mathrlap{\smash{\overbrace{\phantom{%
    \begin{matrix} #2 \end{matrix}}}^{\mbox{$#1$}}}}#2}

\newcommand\coolunder[2]{\mathrlap{\smash{\underbrace{\phantom{%
    \begin{matrix} #2 \end{matrix}}}_{\mbox{$#1$}}}}#2}

\newcommand\coolleftbrace[2]{%
#1\left\{\vphantom{\begin{matrix} #2 \end{matrix}}\right.}

\newcommand\coolrightbrace[2]{%
\left.\vphantom{\begin{matrix} #1 \end{matrix}}\right\}#2}

Under the \emph{out-of-node checksum storage}, the vertical checksums $C = G A$ are distributed into $f \times p_r$ checksum nodes, and each $\frac{n}{p_r} \times \frac{n}{p_c}$ block $C_{i, j}$ is owned by the checksum processor $\Pi_C(i, j)$. Furthermore, by \eqref{generator_mat1}, each checksum $C_{i,j}$ can be computed as:

\begin{align}
\label{checksum_overhead1}
    C_{i,j} = \sum_{t=0}^{p_r-1} \tilg_{i,t} A_{t, j}.
\end{align}
The following Theorem  shows that the random checksum generator matrix constructed in Construction~\ref{const:G_formula1}  satisfies the MDS condition with probability 1.
\begin{theorem}\label{thm:prob_det}
Every square submatrix of the matrix $\widetilde{G}$ constructed following Construction~\ref{const:G_formula1} is full rank with probability 1. 
\end{theorem}
The proof of Theorem \ref{thm:prob_det} and an example for the code construction are respectively deferred to Appendix \ref{mds_theorem_proof} and \ref{construction_example_appen}. 
We further empirically validate Theorem \ref{thm:prob_det}, by comparing the smallest determinant and highest condition number of all square submatrices of $\widetilde{G}= \begin{bmatrix} \widetilde{G}_1 & \widetilde{V} \end{bmatrix}$ following Construction \ref{const:G_formula1} (i.e. random $\widetilde{V}$ and  $\widetilde{G}_1 = -\frac{1}{2} \widetilde{V} \widetilde{V}^T$ ) versus that of the completely random matrix $\widetilde{G}' = \begin{bmatrix} \widetilde{G}_1' & \widetilde{V} \end{bmatrix}$ where entries of $\widetilde{G}_1'$  are also drawn iid from $\mathrm{Unif}(0,1)$.
We compare the least determinants and highest condition numbers of these 2 matrices for $f= 2, \ldots, 6$ and $p_r = 16, 32, 64$ in Figure \ref{fig: rand}. The results in the plot are the average of 100 trials\footnote{Just for $(f =6, p_r =64)$, the number of trials was 10 due to the long running time.}. The results show that the least determinant of square submatrices of the structured matrix $\widetilde{G}$ is comparable to that of the completely random matrix $\widetilde{G}'$, and the highest condition number of square submatrices of the structured matrix $\widetilde{G}$ is at most a factor of $3$ from that of the completely random matrix $\widetilde{G}'$. 
% As it is well-known that the random matrix $\widetilde{G}' $ is MDS with probability 1 , this empirically asserts that our code construction has the MDS property with probability 1. 

\begin{figure}[t]
\centering
\begin{subfigure}{0.43\textwidth}
\includegraphics[width=\textwidth]{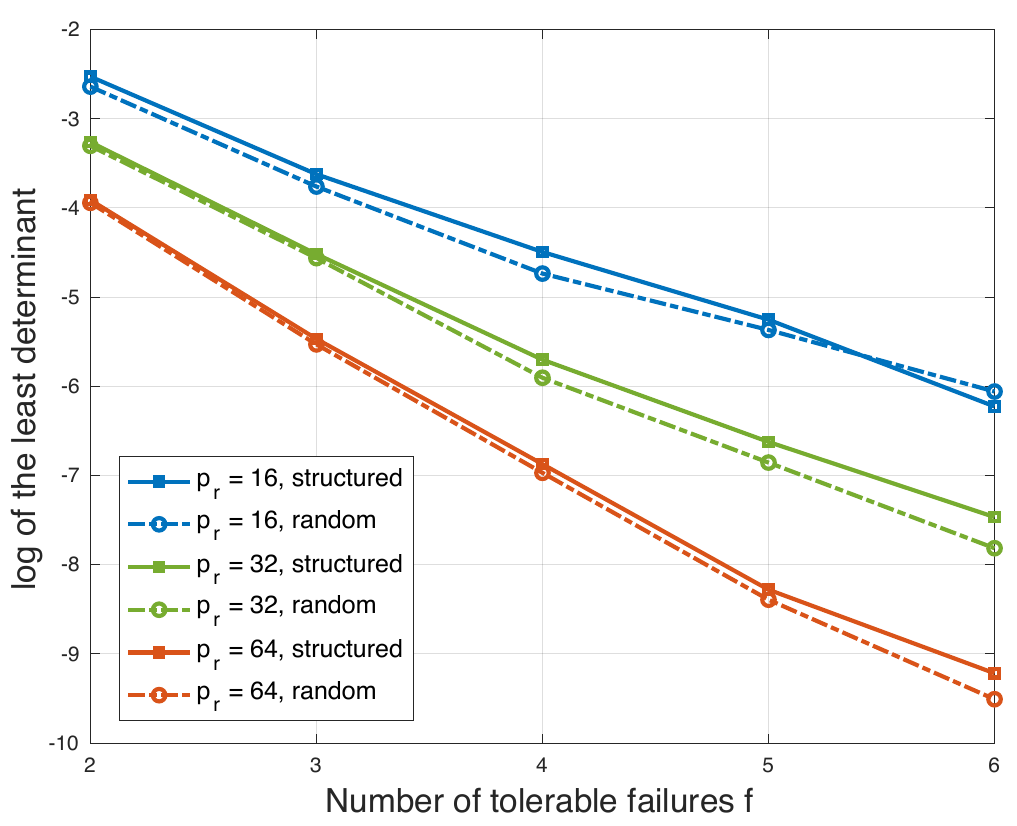}
\caption{Log of the least determinant of square submatrices.}
\label{fig: experiment}
\end{subfigure}
\qquad \qquad
\begin{subfigure}{0.43 \textwidth}
\includegraphics[width=\textwidth]{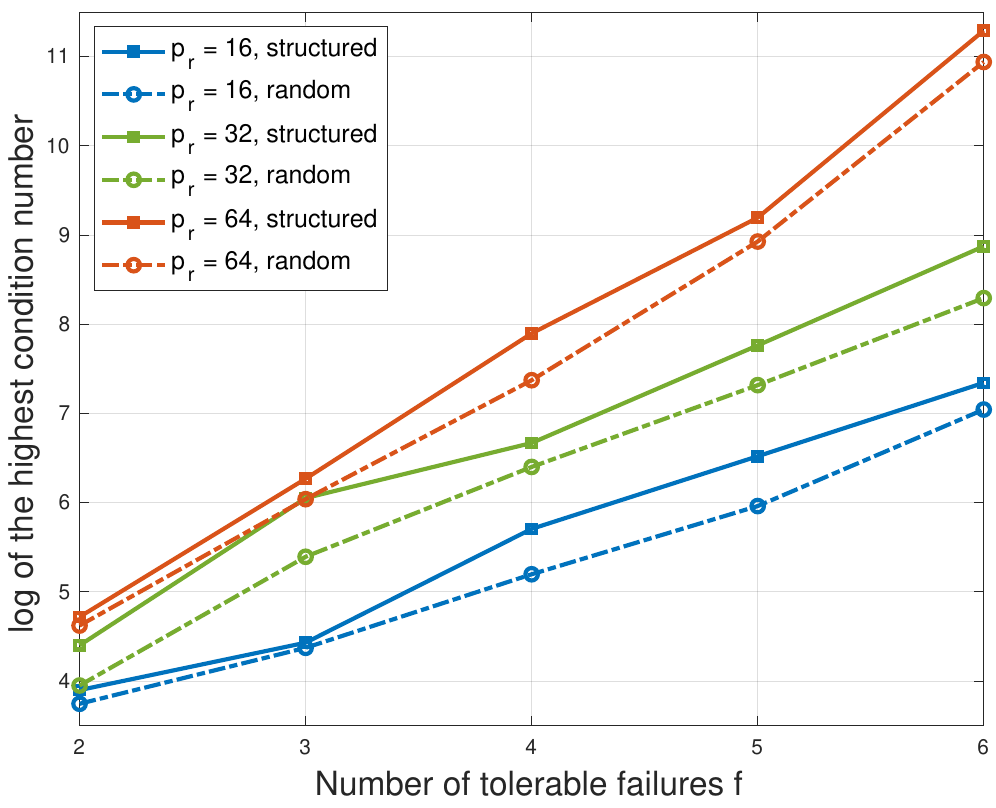}
\caption{Log of the highest condition number of square submatrices.}
\label{fig: condition_number}
\end{subfigure}
\caption{ \small
Comparison between the structured matrix $\widetilde{G}$ and the completely random matrix $\widetilde{G}'$, both of which have size $f \times p_r$, in terms of the least determinant of square submatrices (Figure \ref{fig: experiment}) and the highest condition number of square submatrices (Figure \ref{fig: condition_number}). We report the results for $f= 2, \ldots, 6$ and $p_r = 16, 32, 64$. Each color (blue, green, red) corresponds to each value of $p_r$ ($=16, 32, 64$) being used for comparing $\widetilde{G}$ versus $\widetilde{G}'$. Solid lines are for the structured matrices and dotted lines are for completely random matrices. }
\label{fig: rand}
\end{figure}

\subsection{Checksum-Generator Matrix for In-Node Checksum Storage} 
\label{q_in_node}

We highlight that the above results hold for the general in-node checksum storage setting. Next, we demonstrate the application of our in-node systematic MDS coding to the design of coded QR decomposition.
For clarity, we now turn back to differentiating between the vertical checksum generator matrix $G_v$ and the horizontal checksum generator matrix $G_h$ as in \eqref{eq:encoding}, and recall that the encoded input matrix $\widetilde{A}$ can be written as: 
 % Section \ref{sys_model:checksum_generation}
\begin{equation*}
 \widetilde{A} = \begin{bmatrix}
	 A & A G_h\\
	 G_v A & G_v A G_h
\end{bmatrix},
\end{equation*}
in which, by Lemma \ref{lem:2way_checksum},  $G_v$ 
 and $G_h$ respectively translate to $Q$-factor protection and  $R$-factor protection. Adapting Theorem \ref{thm:in_node_opt} to our setting of 2D block distribution, we have $L= p_r$ data blocks and $P= p_r = p_c$, and consequently consider $K = f + f \ceil{\frac{f}{P-f}}$ that achieves equality in \eqref{eq:in_node_main}. Then the design of horizontal checksum generator matrix $G_h$ can follow immediately Construction \ref{const:in_node}. On the other hand, such direct adoption for $Q$-factor protection is not feasible, since as discussed in Section \ref{subsec:general_scheme}  the vertical checksum generator matrix $G_v$ must satisfy the post-orthogonalization condition \eqref{eq:restrictionG} to circumvent the degraded orthogonality in the decoding phase. Leveraging our idea of semi-random checksum design in Construction \ref{const:G_formula1} and the fact  that the MDS coding is  applicable to in-node setting with negligible additional cost (Remark \ref{rem:alt_in_node}), we propose the following vertical checksum generator matrix $G_v$ that  with high probability  is MDS and thus tolerable to $f$ node failures.

% \begin{construction}[In-Node Systematic MDS Coding] \label{const:in_node}
% Let $D_1, \ldots, D_L$ be data blocks and $C_1, \ldots C_K$ be checksum blocks, and let $P$ be the number of nodes. We assume that $P$ divides $L$ and we assume the load-balanced setting, i.e, each node has $L/P$ data blocks and the checksum distribution follows the load-balanced in-node checksum distribution. Let $\mathcal{I}_p$ be the set of data block indices in the $p$-th node and let $\mathsf{Proc}(D_i)$  or $\mathsf{Proc}(C_i)$  denote the processor index the data or checksum block belongs to. 
% Let $\tilG \in \bbR^{K \times \frac{(P-1)L}{P}}$ be a checksum-generating matrix of a $(K+(P-1)\frac{L}{P}, (P-1)\frac{L}{P})$ systematic MDS code. Then, In-Node Systematic MDS Coding encodes the $i$-th checksum block $C_i$ as follows:
% \begin{equation} \label{eq:in_node_cksum}
%     C_i = \sum_{j \notin \mathcal{I}_{\mathsf{Proc}(C_i)}} \tilG_{i,j} D_j.
% \end{equation}
% \end{construction}

\begin{construction}[In-Node $Q$-factor  Checksum-Generator Matrix for Multiple-Node
Failure Recovery]
\label{const:in_node_qr}
Let $\widetilde{G} = \begin{bmatrix} \widetilde{G}_1 & \widetilde{V} \end{bmatrix} = \{ \tilg_{i,j}\}$ be a $K \times p_r$ matrix, where $\widetilde{V}$ is a $K \times (p_r-K)$ matrix whose entries are chosen randomly from $\mathrm{Unif}(0,1)$ and $\widetilde{G}_1 = -\frac{1}{2} \widetilde{V} \widetilde{V}^T$ is an $f\times f$ matrix. The vertical checksum generator matrix is constructed as $G_v =\widetilde{G} \otimes I_{n/p_r}$. 
\end{construction}
\begin{theorem}\label{thm:prob_det_1}
Every square submatrix of the matrix $\widetilde{G}$ constructed following Construction~\ref{const:in_node_qr} is full rank with probability 1. 
\end{theorem}
The proof of Theorem \ref{thm:prob_det_1} can be done similarly to that of Theorem \ref{thm:prob_det}.

\section{Overhead Analysis of the $Q$-factor Protection Scheme}

% Finally, we show through careful analysis that the overhead of fault tolerance including encoding, failure recovery, and low-cost post-orthogonalization is negligible. 
As discussed in Section \ref{in_node}, while there have been optimal coding designs and overhead analysis for $R$-factor protection schemes, the literature for $Q$-factor protection schemes via coding theory remains nascent, as demonstrated in the summary of previous work in Table \ref{table_of_summary}.  This motivated our proposal of our novel coding scheme for $Q$-factor protection in Section \ref{sec:q_factor_protection}. Nevertheless, even if a coding scheme can provide a great failure resilience, if it has a prohibitive computational cost, it would not be a practical solution. In this section, we show through a rigorous analysis that the overhead of coding for $Q$-factor protection is negligible compared to the cost of original computation. We analyze the overhead in terms of time complexity that accounts for both communication and computation. A detailed cost model is given in the following section.

\subsection{Cost Model and MPI collective communication operations}
\label{MPI_oper}

We use the $\alpha$-$\beta$-$\gamma$ model~\cite{MPI_collective1, MPI_collective2} 
to analyze communication and computation costs:
\begin{align}\label{eq:alpha_beta_gamma}
T = \alpha C_1 + \beta C_2 + \gamma C_3 ,
\end{align}
where $C_1$ is the number of communication rounds, $C_2$ is the number of bytes communicated on the critical path, and $C_3$ is the number of floating point operations (flops). The $\alpha$ and $\beta$ terms model
the communication latency and bandwidth, respectively, and the $\gamma$ term is the computation cost.

Both parallel QR decomposition and our coding strategy are implemented via MPI collective operations, which are common communication-computation patterns in parallel computing. Our overhead analysis  depends on the costs of these operations including broadcast,  reduce and all-reduce.  We  denote by $T_{broadcast} (p, w), T_{reduce} (p, w)$ and $ T_{allreduce} (p, w) $ respectively   the cost of MPI\_Broadcast($p, w$), MPI\_Reduce($p, w$), and MPI\_Allreduce($p, w$) for $p$ processors to transfer an array of $w$ words \footnote{It can be other units, such as bytes, but then the term $\beta$ is the bandwidth per bytes. }. We consider the algorithms given in  \cite{optimal_broadcast} for broadcasts, and \cite{MPI_collective2} for reduce and all-reduce, which are optimized for communication of long messages. Illustrations and costs of these MPI operations are given in Table \ref{tab:overhead_analysis}.

\begin{table}[t]
\caption{MPI operations }
\label{tab:overhead_analysis}
\centering
\begin{tabularx}{1\textwidth} { 
  | >{\centering\arraybackslash}X 
  | >{\centering\arraybackslash}X 
  | >{\centering\arraybackslash}X | }
 \hline 
 \textbf{MPI\_Broadcast($p, w$)} & \textbf{MPI\_Reduce($p, w$)} & \textbf{MPI\_Allreduce($p, w$)} \\
 \hline
  \includegraphics[width=0.3\textwidth]{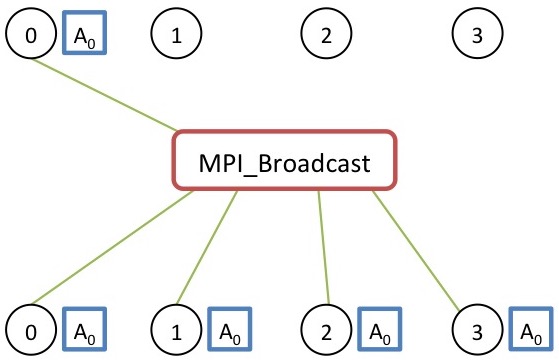} &  \includegraphics[width=0.3\textwidth]{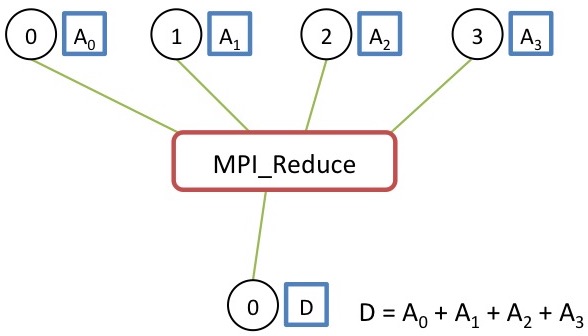}  & \includegraphics[width=0.3\textwidth]{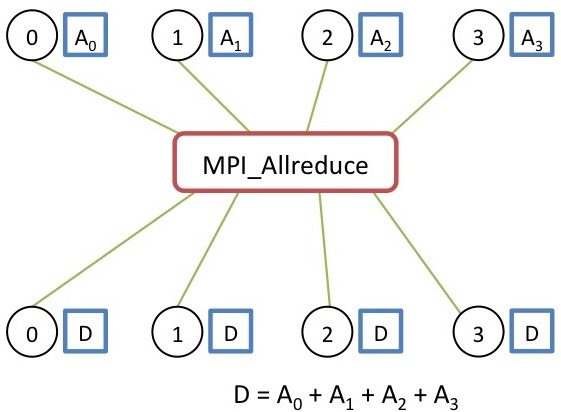} \\
\hline
$T_{broadcast} (p, w)=(\sqrt{\alpha  (\lceil \log p \rceil -1)} + \sqrt{\beta  w})^2$  &  $T_{reduce} (p, w) = 2\alpha \log p + 2\beta \frac{p-1}{p} w +\gamma \frac{p-1}{p} w$ &  $T_{allreduce} (p, w) = 2\alpha \log p + 2\beta \frac{p-1}{p} w +\gamma \frac{p-1}{p} w$  \\
\hline
\end{tabularx}
\end{table}

%%%%%%%%%%%%%%%%%%%%%%%%%%%%%%%%%%%%%%%%%%%%%%

\subsection{Overhead Analysis}
\label{overhead_analysis}
\vspace*{-.1cm}

Two types of overhead are considered: the total overhead of coding $T_{coding}$ and  the overhead for recovering   multiple-nodes failure  $T_{recov}$. The coding overhead $T_{coding}$ is modeled as:
\begin{align} \label{eq:T_coding}
T_{coding} = T_{enc}+ T_{post} +T_{comp}
\end{align}
where $T_{enc}$,  $T_{post}$ and $T_{comp}$ are the overhead for encoding, post-orthogonalization, and increased computation cost for QR factorization of the encoded matrix. We compare $T_{coding}$ with the cost $T_{QR}$ of QR factorization without coding for an $n \times n$ matrix.

% Finally, we analyze the overhead of the proposed coding strategy for $Q$-factor protection. We provide $T_\text{coding}$ in Theorem~\ref{overhead0} and $T_\text{recov}$ in Theorem~\ref{recovery_overhead} by constructing explicit computation and communication strategies for our coding scheme. Combining these two theorems, we show that the overall overhead of the proposed $Q$-factor protection strategy has negligible overhead in Theorem~\ref{thm:overhead}.
% % As the encoding and the retrieval of the uncoded output matrix  respectively are preprocessing and postprocessing of the QR decomposition, and the recovery depends mainly on the generator matrix, the following overheads $T_{enc}, T_{post}$ and $T_{recov}$ are independent of the algorithm used. They only depend on the pattern of 2D block-cyclic distribution and the generator matrix, which is illustrated by the following Theorem \ref{overhead1}.

The following Theorem, whose proof is given in Appendix \ref{overhead_theorem_main1}, summarizes the overhead analysis:

\begin{theorem}
\label{thm:overhead}
Let $T_{QR}$ be the overall cost of the the MGS algorithm~\cite{GS2} for factorizing an $n$-by-$n$ matrix on a $p_r$-by-$p_c$ processor grid. Then, the overhead of the $Q$-factor protection scheme given in Construction~\ref{const:G_formula1} is given as:
\begin{align*}
T_{coding} = O\left(\frac{f}{p_r} + \frac{f}{p_c}\right) \cdot T_{QR}, \;\;\; T_{recov} = O\left(\frac{f}{p_c}\right) \cdot T_{QR},
\end{align*}
where $f$ is the number of failures. Especially,
\begin{align}\label{eq:alpha_cost}
    T^{\alpha}_{coding} = O\left(\frac{f}{n}\right) \cdot T^{\alpha}_{QR}.
\end{align}
\end{theorem}

This Theorem shows that as $p_r$ and $p_c$ grow, the overhead of coding becomes negligible compared to $T_\text{QR}$, as we assume that $f$ is a small fixed constant. Also, \eqref{eq:alpha_cost} shows that the $\alpha$ cost, i.e., the latency cost which is usually orders of magnitude higher than others, scales as $O(\frac{1}{n})$.

% \textcolor{red}{[TODO] Move the below lemmas to the appendix.}
% \begin{lemma}
% \label{overhead0}
% Under the out-of-node checksum storage, the coding strategy in Construction \ref{const:G_formula1} can achieve the following coding overheads:
% \small{
% \begin{align}
% %T_d &=T_{all-to-all} (p_r, \frac{n^2}{P} +\frac{n}{P}) + \gamma (\frac{2n^2}{P}+\frac{2n}{P})\\
% &T_\text{enc} =    f \cdot T_\text{lin-comb}(p_r,\frac{n^2}{P}) + T_{allreduce} (p_r+f, f^2)+ \gamma f^2,  \\
% &T_\text{post} = f \cdot T_\text{broadcast}(p_r, \frac{n(n+p_c)}{P}) + f \cdot T_\text{reduce}(p_r-f+1, \frac{n(n+p_c)}{P}) +
%  \gamma  (2f-1) \frac{n(n+p_c)}{P}, \\
%  & T_\text{comp} \leq  \frac{f}{p_r} T_\text{QR}. 
% \end{align}
% }
% \end{lemma}

% \begin{lemma}
% \label{recovery_overhead} 
% Under the out-of-node checksum storage, the coding strategy in Construction \ref{const:G_formula1} can achieve the following recovery overhead:
% \begin{align}
%     & T_\text{recov} = f \cdot T_{reduce}(p_r+1, \frac{n^2}{P}) + \gamma (f_1^2  + \frac{2}{3} f_1^3+ \frac{fn^2}{P} ).
% \end{align}
% \end{lemma}

%% file: experiment.tex
% \begin{itemize}
%     \item experiment configuration (machine specs, mpi library, mkl) 
%     \item error patterns (diagonal, every iteration)
%     \item any difference from our algorithm description 
% \end{itemize}

\begin{figure}[h]
    \centering
    \includegraphics[width=0.5\linewidth]{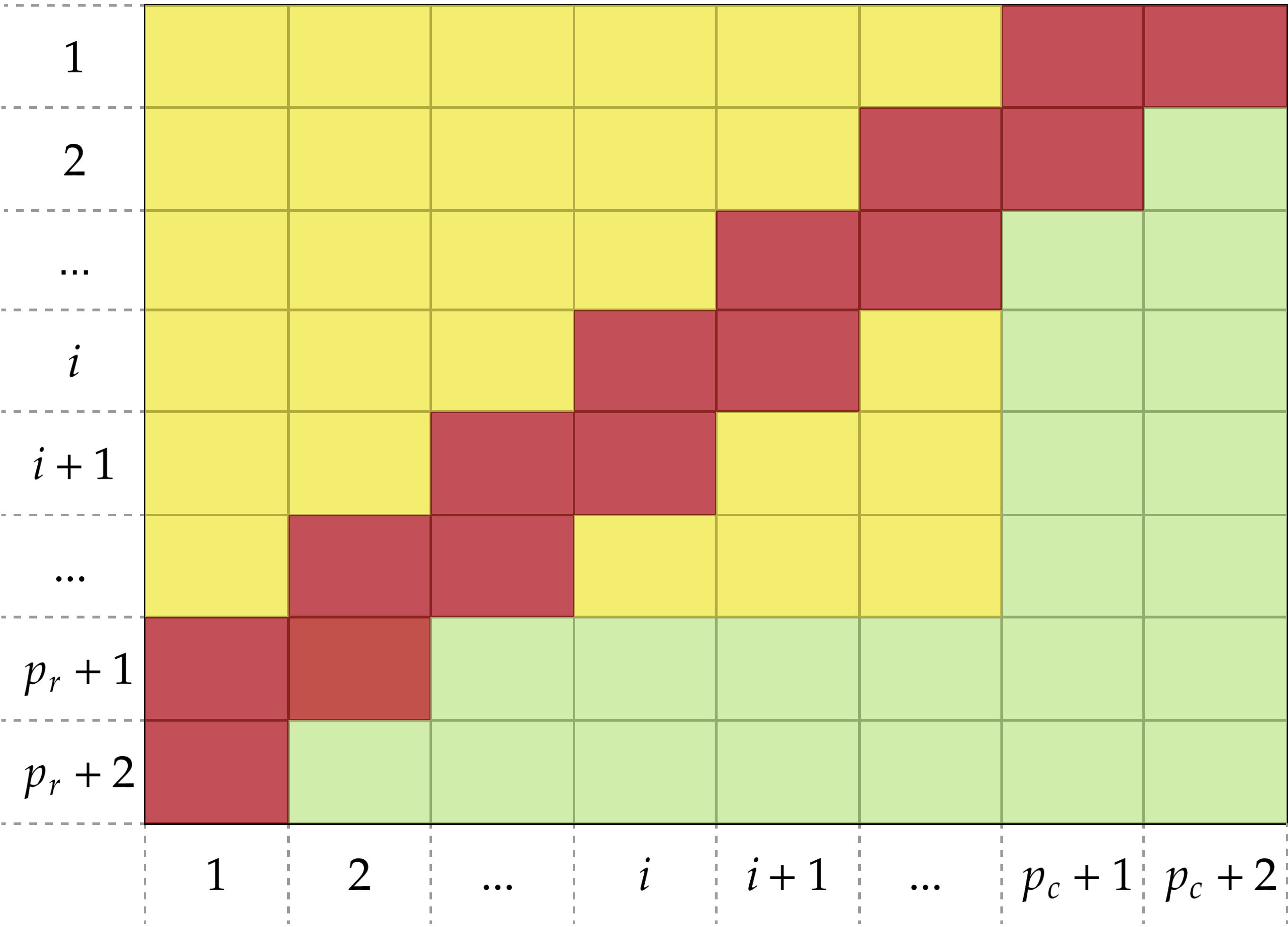}
    \caption{Reverse-diagonal failure pattern for $f=2$.}
    \label{fig:fail-pattern}
\end{figure}

In this section, we empirically demonstrate the negligible overhead of our coded computing framework for fault-tolerant QR decomposition.   

\subsection{Implementation and Experimental Setting}

For practicality, we have a minor deviation from the PBMGS (Algorithm~\ref{alg: PBMGS}) considered in Section~\ref{sec:pbmgs}. 
Specifically, we replace the ICGS algorithm therein with the MGS algorithm for further practical acceleration with negligible error, on the order of $10^{-9}$ on average, as measured by $||A-QR||_2$. Moreover, we make the following considerations in our implementation and cost measurement of the post-orthogonalization step. First, as solving either $ R\bm{x} = Q_1^T  \bm{b}$ or $ R\bm{x} = (G_0 Q_1)^T (G_0 \bm{b})$ (as done in our coded computing framework in Section~\ref{subsec:general_scheme}) requires one call of the triangular solve operation on the original problem size and does not reflect additional cost due to the coding framework, we neglect this step from the cost measurement of the time $T_{QR}$ for a full QR decomposition. Yet, we still take into account the overhead for computing $(G_0 Q_1)$ and $(G_0 \bm{b})$ in the post-orthogonalization overhead $T_{post}$. Second, for the computation of $(G_0 Q_1)$ and $(G_0 \bm{b})$, instead of deploying the optimized procedure in Appendix~\ref{overhead_theorem_main1} that utilizes the sparsity of $G_0$, we directly use the parallel matrix-matrix multiplication routine and empirically demonstrate that the total overhead remains negligible. 

All the experiments were performed on the CSC Pod cluster. Each batch used six nodes containing dual 20-core Intel 6148s, 160 GB RAM, and an OmniPath interconnect. At the software level, Intel MKL 2022.0.2 was used for BLAS and LAPACKE routines, random number generation, and memory management. Intel MPI 2021.5.1 was used for message passing, and batch jobs were submitted through SLURM.

Run-times are measured per step and do not include MPI initialization, input matrix generation, or floating point error calculations.
We report the total overhead $T_{coding}$ of our coding framework 
% as well as its breakdown ($T_{enc}, T_{post}$ and $T_{comp}$), 
and the recovery overhead $T_{recov}$ in the form of the percentage increase with respect to the total running time $T_{QR}$ of parallel QR decomposition on a $n \times n$ matrix. All the results are averaged over 10 trials\footnote{The plots also include standard deviation of every point, which may not be visible due to negligible variation.}. 

\textbf{Setups and failure modeling:} In all tests, we fix $n = 24000$ and vary the number of processors via $p=p_r=p_c \in \{6, 8, 10, 12\}$ as well as the number of tolerable failures (per row/column) $f \in \{0, 1, 2, 3\}$. Note that $f=0$ corresponds to a vanilla run of parallel QR decomposition without the coding scheme. 
For failure modeling, we inject fail-stop errors in every iteration of the PBMGS algorithm. In particular, at the start of every iteration of PBMGS, we simulate the failures of at most $f$ processors on each row and each column of the processor grid in a reverse diagonal pattern as shown in Figure~\ref{fig:fail-pattern}. This results in a total of $O(f p)$ processor failures for a given iteration. Consequently, in every iteration, the coded QR decomposition framework recovers all the failed processors via checksums before proceeding with the computation.

\begin{figure}[t]
    \centering
    \begin{subfigure}[t]{0.30\linewidth}
        \centering
        \includegraphics[width=1.0\linewidth]{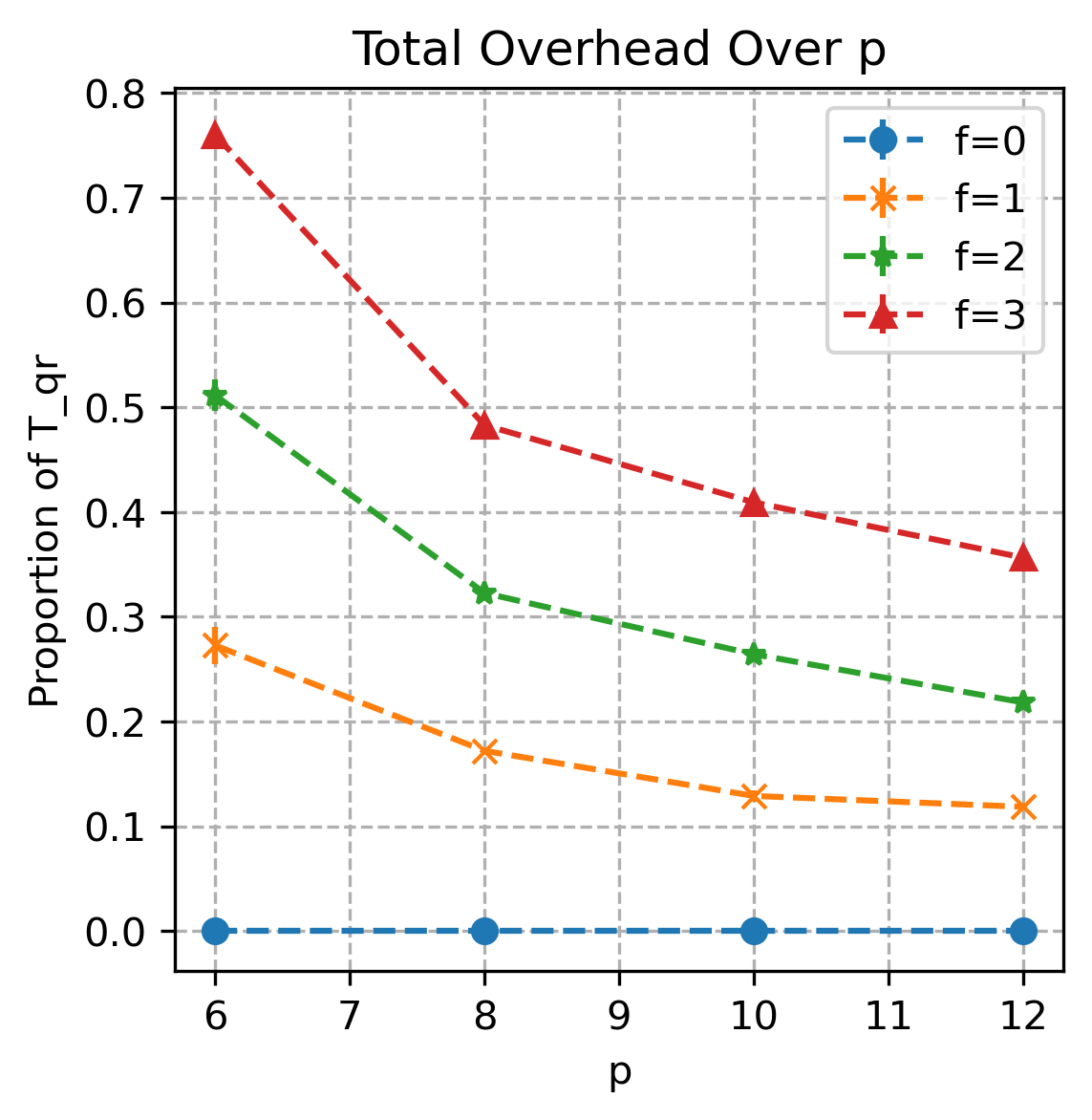}
        \caption{Strong scaling of $T_{coding}$ with $p=p_r=p_c$}
        \label{fig:a}
    \end{subfigure}
    \hfill
    \begin{subfigure}[t]{0.30\linewidth}
        \centering
        \includegraphics[ width=1.0\linewidth]{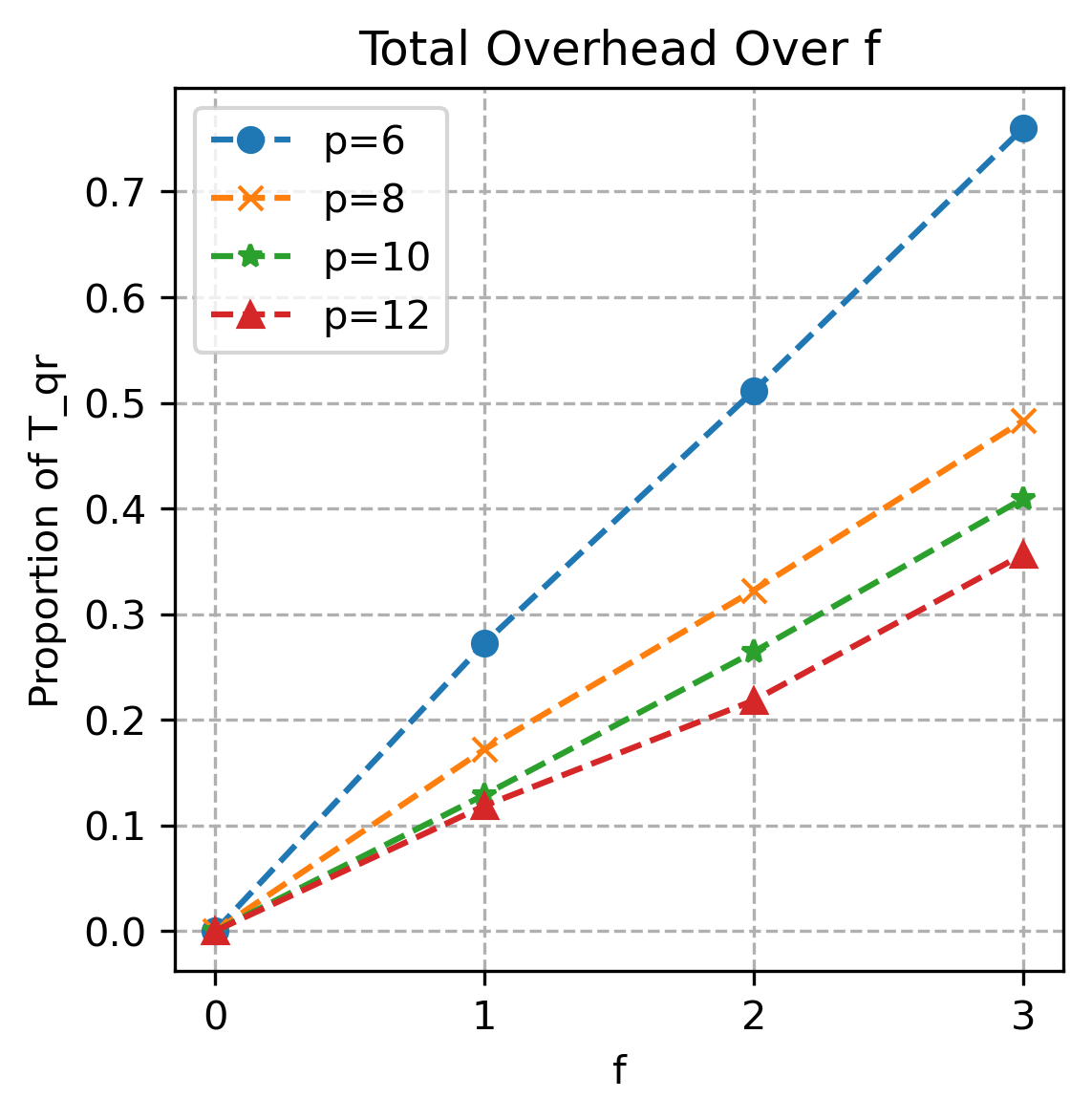}
        \caption{Linear relationship between $T_{coding}$  and~$f$.}
        \label{fig:b}
    \end{subfigure}
    \hfill
    \begin{subfigure}[t]{0.335\linewidth}
        \centering
        \includegraphics[width=1.0\linewidth]{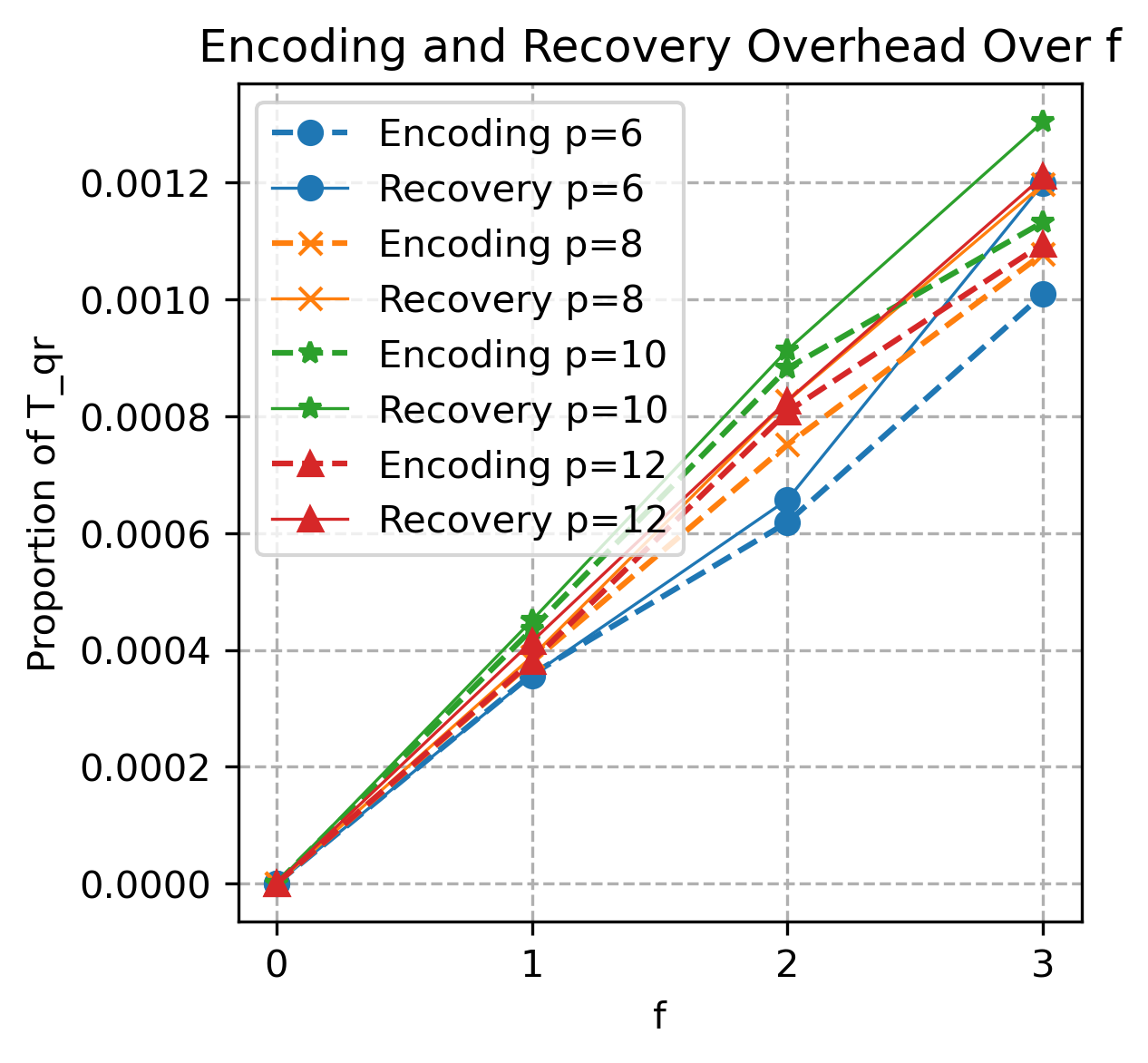}
        \caption{Linear relationship between breakdown overheads (encoding cost $T_{enc}$ and total recovery cost $p \cdot T_{recov}$)  and $f$.}
        \label{fig:c}
    \end{subfigure}
    \caption{Overhead Scaling for Coded QR Decomposition.}
    \label{fig:experiment-plots-ab}
\end{figure}

\subsection{Overhead Analysis}

We evaluate the total overhead $T_{coding}$ incurred by our coding scheme with respect to $p$ and $f$, respectively in Figure \ref{fig:a} and \ref{fig:b}. The results empirically validate that our coding overhead (taken as a proportion of the total QR decomposition cost $T_{QR}$) scales inversely with $p$ (Figure \ref{fig:a}) and linearly with $f$ (Figure \ref{fig:b}), thereby confirming Theorem \ref{thm:overhead}. Moreover, the linear trend in Figure \ref{fig:b} persists across various processor grid sizes $p$, providing evidence that the complexity of our fault tolerance mechanism scales predictably with the degree of redundancy. 

Besides requiring less amount of redundancies for fault-tolerance than checkpointing \cite{BOSILCA2009410}, ABFT is known to support more efficient recovery  \cite{colchecksum}. 
We measure the total recovery overhead of $p \cdot T_{recov} $  across all failures throughout the $p$ iterations of the QR decomposition, each of which incurs $T_{recov}$ for recovering in parallel all the failures within such iteration. We report both the total recovery overhead and the encoding overhead $T_{enc}$ in Figure \ref{fig:c} for different degrees of failures $f$ given a fixed number of processors $p$. The total overhead $p \cdot T_{recov} $ exhibits linear correlation with $f$ and boundedness independent from $p$, thereby implying the growth rate of $O(\frac{f}{p})$ with respect to $T_{QR}$ as suggested by  Theorem \ref{thm:overhead}. Moreover, we note that as the total overhead matches the encoding overhead $T_{enc}$, the actual recovery cost per failure of our coding scheme is minimal.

The empirical findings validate the efficiency of our coding framework in terms of both the negligible overhead required for maintaining fault-tolerance and recovery time, all of which result in increasing uptime for large-scale HPC systems.

%% file: conclusion.tex
%It only works for square systems.. maybe also mention experimental validation?  maybe extension to other algorithms? e.g., comm-efficient householder

In this paper, we extend our prior work \cite{quangQR}  and propose a novel coding strategy for parallel QR decomposition that tolerates multiple-node failures with negligible overhead under both of the out-of-node and in-node checksum storage settings.  
% This paper is motivated by \cite{colchecksum}, and improves upon their work in two fronts. 
For $R$-factor protection, while the optimal coding schemes for out-of-node checksum storage have been presented in the literature, we improve the in-node checksum storage codes in \cite{colchecksum}, by first proving a fundamental bound and proposing a strategy that achieves such lower bound and outperforms the existing strategy~\cite{colchecksum} by $\sim 2$x. 
We inherit realistic assumptions for the HPC computing environment from \cite{colchecksum} and bring innovations from the coding-theoretic perspective.
For $Q$-factor protection, which has only relied on checkpointing-type techniques as the straightforward application of coding can break the orthogonality of the $Q$-factor, we propose a coding strategy with low-cost post-orthogonalization that can recover an orthogonal matrix at the end. 
Extensive experiments across different scales of problem sizes, number of processors, and failures validate our theories and demonstrate the negligible overhead of our coding scheme.
We believe that this work would open up exciting future directions at the interplay between practical HPC algorithms and coded computing.
% We believe that there are a lot of exciting future directions that marry practical HPC algorithms and coded computing, and our work is just one example. 
% In the paper, we study coded computing strategies for parallel QR decomposition and amount of redundancies required for  \emph{in-node checksum storage} in HPC. Specifically, while $R$ factor protection is straightforward, we design the code for $Q$ factor protection against multiple-node failures that addresses the challenge of degraded orthogonality and proves that it satisfies the MDS property with probability 1. We also derive the minimal number of checksums required for \emph{in-node checksum storage} to tolerate any $f$ failures and
% proposes a in-node systematic MDS coding scheme that meets the bound. 

% While we theoretically proved that our coding strategies incur negligible overhead compared to the cost of QR decomposition, an empirical study that measures the overhead on a large-scale HPC system can add more practical value to our work. 
Since our scope in this work is limited to full-rank square matrix $A$, coded QR decomposition for singular matrix or a rectangular matrix $A$  remains as an open question. Also, designing coding strategies for other commonly used algorithms for parallel QR decomposition, such as Householder Transformation~\cite{blockhouseholder1, TSQR_Householder} and Givens Rotation~\cite{givens_rotation2}, would be an interesting direction for future work.

% \section*{Overview}

% \textcolor{red}{[TODO] Catch repeat citations}

% \textcolor{red}{[TODO] Edit the author list and keywords}

\section*{Acknowledgements}
Use was made of computational facilities purchased with funds from the National Science Foundation (CNS-1725797) and administered by the Center for Scientific Computing (CSC). The CSC is supported by the California NanoSystems Institute and the Materials Research Science and Engineering Center (MRSEC; NSF DMR 2308708) at UC Santa Barbara.

%% file: app_in_node.tex
\subsection{Proof of Theorem~\ref{thm:in_node_opt} }
\label{appen_bound_K1}
\begin{proof}
We first prove a lemma that provides a condition on $K$. 
\begin{lemma}\label{lem:in_node_lb}
Under the load-balanced in-node checksum distribution, to recover from any $r$ failed nodes out of $P$ nodes, the number of checksum blocks $K$ must satisfy:
\begin{equation} \label{eq:in_node_lem}
   a \ceil{\frac{K}{P}} + b \floor{\frac{K}{P}} \geq \; f \frac{L}{P}
\end{equation}
for any $a+b = P-f$ and $a \leq \overline{a}$ and $b \leq  \overline{b}$.
\end{lemma}
\begin{proof}[Proof of Lemma~\ref{lem:in_node_lb}]
Assume there exists $a'$ and $b'$ such that $a'+b' = P-f$, $a' \leq \overline{a}$, $b' \leq \overline{b}$ that satisfy $ a' \ceil{\frac{K}{P}} + b' \floor{\frac{K}{P}} < f \frac{L}{P}$.
If $a'$ nodes with $\ceil{\frac{K}{P}}$ checksum blocks and $b'$ nodes with $\floor{\frac{K}{P}}$ checksum blocks are the only surviving nodes, then we have fewer than $f \frac{L}{P}$ from the $P-f$ surviving nodes. However, we have $f \frac{L}{P}$ data blocks missing from the $f$ failed nodes. Since we have fewer equations than the number of variables, we cannot recover all the data blocks (Contradiction).
\end{proof}

By Lemma 1, we have: 
\begin{equation} \label{eq:inf_thr_lb}
  a \ceil{\frac{K}{P}} + b \floor{\frac{K}{P}} \geq \; f \frac{L}{P}
\end{equation}
for any $a+b = P-f$ and $a \leq \overline{a}$ and $b \leq  \overline{b}$.
As this has to be satisfied for the worst case $a$ and $b$, let us consider:
\begin{equation}
    \min_{\substack{a,b \geq 0, \; a+b =f, \\ a \leq \overline{a}, \; b \leq \overline{b}}}  a \ceil{\frac{K}{P}} + b \floor{\frac{K}{P}} = \begin{cases} (P-f) \floor{\frac{K}{P}} & \text{ if } \overline{b} \geq P- f\\ 
    (P-f-\overline{b}) \ceil{\frac{K}{P}} + \overline{b} \floor{\frac{K}{P}}  & \text{ otherwise. }\end{cases}
\end{equation}
Note that for any $K$, $(P-f) \ceil{\frac{K}{P}}  \geq (P-f) \floor{\frac{K}{P}}$ and  $(P-f) \ceil{\frac{K}{P}}  \geq (P-f-\overline{b}) \ceil{\frac{K}{P}} + \overline{b} \floor{\frac{K}{P}} $. Hence, 
\begin{equation}
    (P-f) \ceil{\frac{K}{P}} \geq \min_{\substack{a,b \geq 0, \; a+b =f, \\ a \leq \overline{a}, \; b \leq \overline{b}}}  a \ceil{\frac{K}{P}} + b \floor{\frac{K}{P}}  \geq f \frac{L}{P}.
\end{equation}
Because $\ceil{\frac{K}{P}} $ is an integer,
\begin{equation}
    \ceil{\frac{K}{P}} \geq \ceil{\frac{f L}{(P- f)P}}.
\end{equation}
Now, let us find $K$ such that $\ceil{\frac{K}{P}}  = \ceil{\frac{ f L}{(P-f)P}}$ that satisfy \eqref{eq:inf_thr_lb}. Then, $K$ be expressed as $K = P \ceil{\frac{ f L}{(P-f)P}} - m$ where $m = 0, \ldots, P-1$. We will try to find the largest $m$, and we divide these into three cases: i) $m \geq P-f$, ii) $0 < m < P-f$, and iii) $m = 0$. 

\textbf{Case i)} $m \geq P-f$: Notice that $\overline{b} = m \geq P- f$, and thus we have:
\begin{align}
    \min_{\substack{a,b \geq 0, \; a+b =f, \\ a \leq \overline{a}, \; b \leq \overline{b}}}  a \ceil{\frac{K}{P}} + b \floor{\frac{K}{P}} &=  (P-f) \floor{\frac{K}{P}} =  (P-f) \left( \ceil{\frac{fL}{(P-f)P}}-1 \right). \label{eq:thm_lb_prf1} 
\end{align}
When $\frac{fL}{(P-f)P}$ is an integer, $\text{\eqref{eq:thm_lb_prf1}} = (P-f)\left(\frac{f L}{(P-f)P}-1 \right) < \frac{f L}{P}$.
When $\frac{f L}{(P-f)P}$ is not an integer, $\text{\eqref{eq:thm_lb_prf1}} = (P-f)\floor{\frac{f L}{(P-f)P}} < \frac{f L}{P}$.
Thus, there does not exist $m \geq P-f$ that satisfies \eqref{eq:inf_thr_lb}.

\textbf{Case ii)} $0 < m < P-f$: As  $\overline{b} = m < P- f$, we have:
\begin{align*}
    \min_{\substack{a,b \geq 0, \; a+b =f, \\ a \leq \overline{a}, \; b \leq \overline{b}}}  a \ceil{\frac{K}{P}} + b \floor{\frac{K}{P}} &=  (P-f-\overline{b}) \ceil{\frac{K}{P}} + \overline{b} \floor{\frac{K}{P}} =  K - f \left( \floor{\frac{K}{P}}+1 \right ) \\
    &= K - f \ceil{\frac{K}{P}} = K - f \ceil{\frac{f L}{(P-f)P}} \geq \frac{f L}{P}
\end{align*}
Hence, $K = f \ceil{\frac{f L}{(P-f)P}} +\frac{fL}{P}$ satisfy \eqref{eq:inf_thr_lb}. We can also check that: 
\begin{align*}
 m &= P\ceil{\frac{f L}{(P-f)P}} - K   =P\ceil{\frac{f L}{(P-f)P}} - f\ceil{\frac{f L}{(P-f)P}} - \frac{f L}{P} \\ 
 &= (P-f)\ceil{\frac{f L}{(P-f)P}}-\frac{f L}{P} < P-f.
\end{align*}
As we already found $m$ in \textbf{Case ii)}, we do not have to consider \textbf{Case iii)}. 
\end{proof}

\subsection{Proof of Theorem~\ref{thm:in_node_mds}}
\label{appen_in_node_R1}

\begin{proof} 
From any $P-f$ surviving nodes, we have $(P-f)\frac{L}{P}$ data blocks and $f\frac{L}{P}$ checksum blocks as $K$ satisfies the condition in \eqref{eq:in_node_lem}. Let $p_1, \ldots, p_r$ be the index of failed nodes. After subtracting the data blocks from the survived nodes, the checksum equations given in \eqref{eq:in_node_cksum} reduces to equations in $f\frac{L}{P}$ data blocks as follows: 
\begin{equation}
     C_i = \sum_{j \in \cup_{k=1}^{r} \mathcal{I}_{p_{k}}} \tilG_{i,j} D_j.
\end{equation}
We have $f\frac{L}{P}$ such equations and since any square submatrix of $\tilG$ is full-rank, we can solve the set of linear equations to retrieve $D_j$'s for $j \in \cup_{k=1}^{r} \mathcal{I}_{p_{k}}$.
\end{proof}

\subsection{Proof of Theorem~\ref{thm:G_0}}
\label{appen_post_ortho_proof1}

\begin{proof}
Recall that given $G = \begin{bmatrix}  G_1 & V \end{bmatrix}$, the $n \times n$  matrix $G_0$ is computed by:
\begin{align*}
G_0 =
\begin{bmatrix}
I_{c}+G_1 & V \\
V^T & -I_{n-c}
\end{bmatrix}
\end{align*}

The following Observation \ref{G0_relation} is used in our proof later: 
\begin{observation}
\label{G0_relation}
Under the restriction \eqref{eq:restrictionG} that $G_1 = -\frac{1}{2} V V^T$, the following equation holds:  $G_0^T G_0 = I + G^T G$. 
 \end{observation}
\begin{proof}[Proof of Observation \ref{G0_relation}]
From  \eqref{eq:restrictionG}, we have $G_1^T=G_1=  \frac{-1}{2}  V V^T$, and obtain that:
\begin{align*}
G_0^T G_0 &= \begin{bmatrix}
I+G_1 & V \\
V^T & -I
\end{bmatrix}^T 
\begin{bmatrix}
I+G_1 & V \\
V^T & -I
\end{bmatrix} =
\begin{bmatrix}
(I+G_1)^2 + V V^T & (I+G_1) V - V \\
V^T (I+G_1)-V^T & V^T V + I
\end{bmatrix} \\
&=
\begin{bmatrix}
I+G_1^2 & G_1 V \\
V^T G_1 & V^T V + I
\end{bmatrix} = I +\begin{bmatrix}
G_1^2 & G_1 V \\
V^T G_1 & V^T V
\end{bmatrix} = I + G^T G.
\end{align*} 
\end{proof}
 At the end of QR decomposition of encoded $\widetilde{A}$ in equation  (\ref{encoded_qr_compute}), the left factor $\begin{bmatrix} Q_1 \\ GQ_1\end{bmatrix}$ is orthogonal : 
 \begin{align*}
 I &= \begin{bmatrix} Q_1 \\ GQ_1\end{bmatrix}^T \begin{bmatrix} Q_1 \\ GQ_1\end{bmatrix}= Q_1^T Q_1 + Q_1^T G^T G Q_1 \\
 & = Q_1^T (I + G^T G)Q_1 = Q_1^T G_0^T G_0 Q_1 \text{   (by Observation \ref{G0_relation})} \\
 &= (G_0 Q_1)^T (G_0 Q_1).
 \end{align*}
 Therefore,$G_0 Q_1$  is orthogonal.
The following Observation \ref{detG0} will be useful for computing the determinant of $G_0$.

\begin{observation}
\label{detG0}
Consider matrices $A,B,C,D$ of size $n\times n, n \times m, m\times n, m \times m$ respectively. If D is invertible then:
$$det\begin{bmatrix} A & B \\ C & D \end{bmatrix} = det(D) det(A-BD^{-1}C)$$
\end{observation}

We compute the determinant of $G_0$: 
\begin{align*}
det(G_0) &= det\begin{bmatrix}
I_{c}+G_1 & V\\
V^T & -I_{m-c}
\end{bmatrix} \\
&= det(-I_{m-c}) det((I +G_1) - V(-I)^{-1}  V^T)\\
&= (-1)^{m-c} det(I + G_1 + V V^T)
\end{align*}
where the second line follows Observation \ref{detG0}. 
Under our setting $G_1 = \frac{-1}{2}  V V^T$, the determinant of $G_0$ becomes:
\begin{align*}
det(G_0) =  (-1)^{m-c} det(I  +\frac{1}{2} V V^T).
\end{align*}
Now taking any non-zero column vector $\bm{z}$, we have $\bm{z}^T (I  +\frac{1}{2} V V^T) \bm{z} = \bm{z}^T \bm{z} + \frac{1}{2} (V^T \bm{z})^T V^T\bm{z} \geq \bm{z}^T \bm{z} > 0$. This means that $I  +\frac{1}{2} V V^T$ is positive definite, and $det(I  +\frac{1}{2} V V^T) \neq 0$. 
Therefore, $det(G_0) \neq 0$ and $G_0$ is invertible. 
\end{proof}

\subsection{Proof of Theorem~\ref{thm:prob_det}}
\label{mds_theorem_proof}
\begin{proof}
Let $n_s$ denote the dimension of $S$. We can divide this into two cases: (i) When $S$ is a submatrix of $\widetilde{V}$, and (ii) when $t$ columns of $S$ are from $\widetilde{G}_1$ and $(n_s-t)$ columns are from $\widetilde{V}$ ($1 \leq n_s \leq f$, $0 \leq t \leq n_s$). For the first case, any square submatrix of $\widetilde{V}$ is a random matrix where each entry was chosen iid. Hence, is is full rank with probability 1. Thus, we focus on the second case. 

Let $\mathrm{Perm}(n_s)$ be all permutations of $\{1, \ldots, n_s\}$, and for permutation $\sigma$, $\sigma_i$ is the $i$-th element in the sequence. For example when $n_s = 3$, if $\sigma = [2 \;\; 3 \;\; 1]$, then $\sigma_1 =2, \sigma_2 = 3, \sigma_3=1$. The Leibniz formula for matrix determinant is:
\begin{equation}
    \mathrm{det}(S) = \sum_{\sigma \in \mathrm{Perm}(n_s)} \mathrm{sgn}(\sigma) \prod_{i=1}^{n_s} s_{i, \sigma_i} = \sum_{\sigma \in \mathrm{Perm}(n_s)} \xi(\sigma), \qquad \xi(\sigma) \triangleq \mathrm{sgn}(\sigma) \prod_{i=1}^{n_s} s_{i, \sigma_i}.
\end{equation}
$\mathrm{sgn}(\sigma)$ is $+1$ or $-1$, but since this term is not important for our argument, we refer to~\cite{trefethen1997numerical} for more explanation. 

We now prove that $\mathrm{det}(S)$ is a non-zero polynomial in $v_{i,j}$'s.  Let us consider the term:
\begin{equation}
    \xi (\sigma = [ 1\;\; \ldots \;\; n_s]) = \prod_{i=1}^{n_s} s_{i, i}. 
\end{equation}
For simplicity, let us assume that the $t$ columns from $\widetilde{G}_1$ are $\mbg_1, \ldots \mbg_t$ and the $n_s -t$ columns of $\widetilde{V}$ are $\mbv_{1}, \ldots \mbv_{n_s-t}$. The argument naturally extends to a general case. Now, $\prod_{i=1}^{n_s} s_{i, i}$ can be expanded as:
\begin{align}
    \prod_{i=1}^{n_s} s_{i, i} &= \prod_{i=1}^{t} \tilg_{i,i} \prod_{i=1}^{n_s -t} v_{i, t+i} \\ 
    &= \left(-\frac{1}{2}\right)^t \cdot \left( \prod_{i=1}^t \sum_{j=1}^{p_r} v_{i,j}^2 \right) \cdot \prod_{i=1}^{n_s -t} v_{i, t+i}. \label{eq:det_expsn}
\end{align}
In \eqref{eq:det_expsn}, let us focus on one term: $(v_{1,p_r})^2 (v_{2, p_r-1})^2 \cdots (v_{t,p_r-t+1})^2 \cdot \prod_{l=1}^{n_s -t} v_{l, t+l} \triangleq \psi$. We show that this term does not appear anywhere else in $\xi (\sigma')$ when $\sigma' \neq [ 1\;\; \ldots \;\; n_s] $, and thus this term cannot be canceled out. We prove this by contradiction. 

Assume there exists a $\sigma'$ such that $\psi \in \xi(\sigma')$. First, let us assume that there exists $1 \leq i \leq t$ such that $\sigma'_i = j \neq i$. We get the $v_{i, p_r-i+1}$ term only from $\tilg_{i,j}$ and $\tilg_{k,i}$ ($k \neq i$) as:
\begin{align}
    \tilg_{i,j} &= v_{i,1}v_{j,1} + \cdots + v_{i,p_r-i+1}v_{j,p_r-i+1} + \cdots v_{i,p_r}v_{j,p_r} \\
    \tilg_{k,i} &= v_{k,1}v_{i,1} + \cdots + v_{k,p_r-i+1}v_{i,p_r-i+1} + \cdots v_{k,p_r}v_{i,p_r}.
\end{align}
Then, the only term that has $(v_{i, p_r-i+1})^2$ is the product of the following terms $v_{i,p_r-i+1}v_{j,p_r-i+1} $ and $v_{k,p_r-i+1}v_{i,p_r-i+1}$, i.e.  
\begin{equation}
    v_{i,p_r-i+1}^2 v_{j,p_r-i+1} v_{k,p_r-i+1}.  \label{miniterm1}
\end{equation}
Notice that this term has $v_{j,p_r-i+1} v_{k,p_r-i+1}$. We next proceed to prove that  $v_{j,p_r-i+1} $ and similarly $v_{k,p_r-i+1}$ are not present in $\psi$. 
First, since $j\neq i$, the term $v_{j,p_r-i+1} $ does not appear in $\prod_{l=1}^t (v_{l, p_r-l+1})^2$ of $\psi$. It is left to show that the term $v_{j,p_r-i+1} $  also does not appear in $ \prod_{l=1}^{n_s -t} v_{l, t+l}$. For the sake of contradiction, suppose that $v_{j,p_r-i+1}$ is one of those $  v_{l, t+l}$'s for $ l =1 \to n_s -t$. Then, this is equivalent to having $j\in [1, n_s -t]$ and $p_r - i + 1 = t+j$, which would imply:
\begin{align*}
    n_s - t &\geq  j = p_r -i + 1 -t \\
    n_s &\geq p_r - i + 1 \geq p_r - n_s +1\\
    \therefore 2n_s &\geq p_r +1 > p_r.
\end{align*}
However, the last line is a contradiction since $2n_s \leq 2f \leq p_r$. Therefore, we have proven that $v_{j,p_r-i+1} $ is not present in $\psi$. Similarly, $v_{k,p_r-i+1}$ and thus \eqref{miniterm1} are not present in $\psi$.
Hence, $\psi \notin \xi(\sigma)$.
Now, assume that there exists $t+1 \leq i \leq n_s$ such that $\sigma'_i = j \neq i$. Then, $\xi(\sigma')$ always has the term $v_{t+i, j}$ ($j \neq i$), which was not present in $\psi$. Hence, $\psi \notin \xi(\sigma)$. 
This proves that $\psi$ appears only in $\xi(\sigma = [1 \;\; \cdots \;\; n_s])$ and hence this term has a a non-zero coefficient, $(-1/2)^t$. As $\mathrm{Pr}(\{v_{i,j}\}_{\substack{i=1,\ldots,r, \\ j=1, \ldots, p_r}}: p(\{ v_{i,j}\}) =0) = 0$ for a non-zero polynomial $p$, this completes the proof.
\end{proof}

\subsection{Proof of Theorem~\ref{thm:overhead}}
\label{overhead_theorem_main1}

\subsubsection{Supplementary Lemmas}
First, we present the following Lemmas that characterize the analytical overhead of  coded QR decomposition. 
\begin{lemma}
\label{overhead0}
Under the out-of-node checksum storage, the coding strategy in Construction \ref{const:G_formula1} can achieve the following coding overheads:
\small{
\begin{align}
%T_d &=T_{all-to-all} (p_r, \frac{n^2}{P} +\frac{n}{P}) + \gamma (\frac{2n^2}{P}+\frac{2n}{P})\\
&T_\text{enc} =     f \cdot T_\text{reduce}(p_r+1, \frac{n^2}{P})  + \gamma \frac{fn^2}{P}+ \gamma f (f+1) (p_r-f),  \\
\nonumber
&T_\text{post} = f \cdot T_\text{broadcast}(p_r, \frac{n(n+p_c)}{P}) + f \cdot T_\text{reduce}(p_r-f+1, \frac{n(n+p_c)}{P}) \\
& \quad \quad \quad +
 \gamma  (2f-1) \frac{n(n+p_c)}{P}, \\
 & T_\text{comp} \leq  \frac{f}{p_r} T_\text{QR}. 
\end{align}
}
\end{lemma}

% \subsection{Proof of Lemma  \ref{overhead0}}
\begin{proof}
Now we analyze the overhead of encoding, recovery (decoding), and post-processing.

1) Encoding overhead: For encoding, we instantiate $\widetilde{G} = (\widetilde{g}_{i,j})$ following Construction \ref{const:G_formula1},   and   compute 
$C_{i,j} = \sum_{t=0}^{p_r-1} \tilg_{i,t} A_{t, j}$ 
% for $i= 0, \ldots, f-1$, 
to be sent to the processor $\Pi_C(i, j)$. The former  takes $\gamma f (p_r-f)$ computational cost for generating the $f \times (p_r-f)$ random matrix $\widetilde{V} = [\widetilde{v}_1, ..., \widetilde{v}_{p_r-f}]$ (where $\widetilde{v}_i\in \Br^{f\times 1}$) on all processors\footnote{In this step, due to small computational burden, we can let all processors to generate in parallel the same random matrix $\widetilde{V}$ via the same random seed.} and $\gamma f^2 (p_r-f)$ cost to compute $\widetilde{G}_1 = -\frac{1}{2} \widetilde{V} \widetilde{V}^T$.
% $\gamma f^2 + T_{allreduce} (p_r+f, f^2)$   cost to compute 
% \begin{align}
% \widetilde{G}_1 = -\frac{1}{2} \widetilde{V} \widetilde{V}^T =-\frac{1}{2} \sum_{i=1}^{p_r-f}  \widetilde{v}_i \widetilde{v}_i^T \label{helper_baby1}.    
% \end{align}
% In particular, we  choose any $p_r-f$ out of the total of $p_r + f$ vertical processors to compute in parallel $ \widetilde{v}_i \widetilde{v}_i^T \in \Br^{f\times f}$ (for $i \in [1, p_r-f]$) in $\gamma f^2$ before aggregating them in view of \eqref{helper_baby1} via the MPI\_Allreduce operation; note that while the actually computation is done over $p_r-f$ vertical processors holding the data, we use the MPI\_Allreduce over $p_r+f$ vertical processors to broadcast $\widetilde{G}_1$ to all processors\footnote{For the remaining $2f$ processors, we can initialize them with all-zero $f \times f$ matrices.}.
% % thereby resulting in the total $O(f^2 (p_r-f))$ complexity.
The later step computing $C_{i,j}$ can be done by sequentially performing linear-combination reduce in $f$ iterations. 
Hence, the total encoding overhead is: 
\begin{align*}
    T_\text{enc} &= \gamma f (f+1) (p_r-f)+ f \cdot T_\text{lin-comb}(p_r,\frac{n^2}{P}) \\
    &=  \gamma f (f+1) (p_r-f)+ f \cdot \left[T_\text{reduce}(p_r+1, \frac{n^2}{P}) + \gamma \frac{n^2}{P}\right] \\
    &= f \cdot T_\text{reduce}(p_r+1, \frac{n^2}{P})  + \gamma \frac{fn^2}{P}+ \gamma f (f+1) (p_r-f). 
\end{align*}

2) Post-processing overhead: The overhead for post-processing is attributed to the cost for performing matrix multiplications $(G_0 Q_1)$ and $(G_0 \bm{b})$ (refer to Section \ref{subsec:general_scheme}). 
We hereby present an efficient way for such computation. 
Since $G_0 Q_1 = [ G_0 Q_1[:, 0], Q_1[:, 1], ..., Q_1[:, p_r-1] ]$, we compute  all the $G_0 Q_1[:, j]$'s in parallel. Furthermore, this can be combined with the computation  of $(G_0 \bm{b})$ for efficiency
Specifically, to formally describe our  computational pipeline, we first consider the $n \times (\frac{n}{p_c}+1)$ matrix $\overline{Z}^j = \begin{bmatrix} Q_1[:, j] & \bm{b}\end{bmatrix} = \begin{bmatrix} \overline{Z}^j_{0} \\ \overline{Z}^j_{1} \\ \vdots \\ \overline{Z}^j_{p_r-1} \end{bmatrix}$, where each sub-block $\overline{Z}^j_i$ is of dimension $\frac{n}{p_r} \times (\frac{n}{p_c}+1)$ . 
Then, we aim to compute $G_0 \overline{Z}^0, G_0 \overline{Z}^1, ..., G_0 \overline{Z}^{p_r-1}$ in parallel across the $p_r$ block-rows. 
%Since $G_0 Q_1 = [ G_0 Q_1[:, 0], Q_1[:, 1], ..., Q_1[:, p_r-1] ] $
Furthermore, thanks to parallelisim, the post-processing overhead is the cost for performing matrix multiplication $G_0 \overline{Z}^j$, which is further expressed as the followings to utilize the sparsity of $G_0$:
\begin{align}
\nonumber
G_0 \overline{Z}^j &= \begin{bmatrix}
I_{\frac{fn}{p_r}}+G_1 & V \\
V^T & -I_{n-\frac{fn}{p_r}}
\end{bmatrix} \overline{Z}^j  \\
\nonumber
&= \begin{bmatrix}
H \otimes I_{\frac{n}{p_r}} & \widetilde{V} \otimes I_{\frac{n}{p_r}} \\
\widetilde{V}^T \otimes I_{\frac{n}{p_r}}  & -I_{p_r - f} \otimes I_{\frac{n}{p_r}} 
\end{bmatrix} \begin{bmatrix} \overline{Z}^j_{0} \\ \overline{Z}^j_{1} \\ \vdots \\ \overline{Z}^j_{p_r-1} \end{bmatrix}  \text{ where $H = (I_f+\widetilde{G}_1)$} \\
\nonumber
&= \begin{bmatrix} \sum_{i=0}^{f-1} h_{0,i} \cdot \overline{Z}^j_i + \sum_{t=f}^{p_r-1} \widetilde{v}_{0,t} \cdot  \overline{Z}_t  \\ 
\sum_{i=0}^{f-1} h_{1,i} \cdot \overline{Z}^j_i + \sum_{t=f}^{p_r-1} \widetilde{v}_{1,t} \cdot  \overline{Z}^j_t  \\ 
\vdots \\
\sum_{i=0}^{f-1} h_{f-1,i} \cdot \overline{Z}^j_i + \sum_{t=f}^{p_r-1} \widetilde{v}_{f-1,t} \cdot  \overline{Z}^j_t  \\ 
\sum_{i=0}^{f-1} \widetilde{v}_{i,0} \cdot  \overline{Z}^j_i   - \overline{Z}^j_{f}\\
\sum_{i=0}^{f-1} \widetilde{v}_{i,1} \cdot  \overline{Z}^j_i   - \overline{Z}^j_{f+1}\\
\vdots \\
\sum_{i=0}^{f-1} \widetilde{v}_{i,p_r-f-1} \cdot  \overline{Z}^j_i   - \overline{Z}^j_{p_r-1}\\
\end{bmatrix} 
\triangleq \begin{bmatrix} W_0 \\W_1 \\ \vdots \\ W_{p_r-1} \end{bmatrix}. 
\end{align}
Notice that all $W_i$'s for $i=0, \ldots, p_r-1$ have some linear combination of $\overline{Z}^j_t$ ($t=0, \ldots f-1$). To compute this, we first broadcast $\overline{Z}^j_t$'s ($t=0, \ldots f-1$) to all the nodes. This requires performing $f$ operations of  MPI\_broadcast($p_r$, $\frac{n}{p_r} \times (\frac{n}{p_c}+1)$). 
Next, we describe the steps of computing these $W_i$'s for $i=0, \ldots, p_r-1$ in details for algorithmic implementation in practice.
Before we begin the series of broadcast operations, we first initialize $W_i$ at each node as:
\begin{align*}
    W_i = h_{i,i} \overline{Z}^j_i \;\; \text{ for } i <f,  \;\; W_i = -\overline{Z}^j_i \;\; \text{ for } i \geq f.
\end{align*}
After the $l$-th broadcast of the matrix $\overline{Z}^j_l$ ($l = 0, \ldots, f-1$), nodes will update its $W_i$ as follows: 
\begin{align*}
      W_i = W_i + h_{i,l} \overline{Z}^j_{l} \;\; \text{ for } i <f (i \neq l), \;\; W_i = W_i +\widetilde{v}_{l, i} \overline{Z}^j_{l} \;\; \text{ for } i \geq f.
\end{align*}
The initialization and updating $W_i$ after each iteration requires a total of $(2f-1) \frac{n}{p_r} \times (\frac{n}{p_c}+1) = (2f-1) \frac{n(n+p_c)}{P} $ flops. 
After this, the computation for $W_f, \ldots, W_{p_r-1}$ is complete, and we only need more computation for $W_0, \ldots W_{f-1}$ to obtain $ \sum_{t=f}^{p_r-1} \widetilde{v}_{i,t} \cdot  \overline{Z}^j_t $ for $i=0, \ldots, f-1$. This can be achieved by $f$ operations of linear-combination reduce from the last $p_r-f$ nodes to one of the first $f$ nodes, which incurs $f T_\text{lin-comb} (p_r-f, \frac{n(n+p_c)}{P})$.
% , which should have precomputed $\widetilde{v}_{i,t} \cdot  \overline{Z}_t$ incurring $\gamma \frac{n}{p_r} \cdot(\frac{n}{p_c}+1)$ , 
% to one of the first $f$ nodes, which is equivalent to MPI\_reduce($p_r-f+1$, $\frac{n}{p_r} \times (\frac{n}{p_c}+1)$). 
Hence, the total cost of post-processing is given by: 
\begin{align*}
    T_\text{post} = &f \cdot T_\text{broadcast}(p_r, \frac{n(n+p_c)}{P}) + f \cdot T_\text{reduce}(p_r-f+1, \frac{n(n+p_c)}{P}) \\
    &+
 \gamma  (2f-1) \frac{n(n+p_c)}{P}.
\end{align*}

3) Increased computation cost for QR decomposition of the encoded matrix: Using Observation \ref{overhead2}, we have $ T_\text{comp} \leq \frac{c}{n} T_\text{QR}$ where $ T_\text{comp}$ is the overhead for running QR decomposition on the larger $(m+c) \times n$ encoded matrix $\widetilde{A}$ instead of $m \times n$ matrix $A$. In our code construction, we use $c = \frac{fn}{p_r}$ checksums and thus obtain that:
\begin{align}
    T_\text{comp} &\leq \frac{fn/p_r}{n} T_\text{QR} = \frac{f}{p_r} T_\text{QR} \label{comp_cost1}.
\end{align}
\end{proof}

\begin{lemma}
\label{recovery_overhead} 
Under the out-of-node checksum storage, the coding strategy in Construction \ref{const:G_formula1} can achieve the following recovery overhead:
\begin{align}
    & T_\text{recov} = f \cdot T_{reduce}(p_r+1, \frac{n^2}{P}) + \gamma (f_1^2  + \frac{2}{3} f_1^3+ \frac{fn^2}{P} ).
\end{align}
\end{lemma}

% \subsection{Proof of Lemma \ref{recovery_overhead}}
\begin{proof}
Before we begin our analysis we define a commonly-occurring operation, \textit{``linear-combination reduce''}. The goal of this operation is to compute the linear combination of the data blocks at $p$ nodes and then send it to a separate destination node. I.e., we compute $\alpha_1 D_1 + \cdots \alpha_p D_p$ for some $\alpha_i \in \bbR \;\; (i=1, \ldots, p)$ distributedly, and send it to the $(p+1)$-th node. This can be achieved through  MPI\_reduce($p+1$, $w$) operation by setting the data block at the destination node to be $\bm{0}$. The overhead of this operation when the size of the data blocks is $w$ is: $T_\text{lin-com}(p, w) = T_\text{reduce}(p+1, w) + \gamma w $ where $\gamma w$ comes from computing  $\alpha_i \cdot D_i$ for $i=1,\ldots, p$ in parallel. 

Recovery requires first decoding the lost data blocks from the failed systematic nodes and then recomputing the lost checksum blocks from the failed checksum nodes. Let $f_1$ and $f_2$ be the number of failed systematic nodes and failed checksum nodes. In the proof we only consider the case when $f = f_1 + f_2$ because the overhead when the number of failures is smaller than $f$ is upper bounded by the worst-case recovery cost of having $f$ failures. 

We partition the index set  $\{0, 1, ..., p_r-1\}= \mathcal{A}_\text{fail} \cup \mathcal{A}_\text{succ}$ , where $\mathcal{A}_\text{fail}$ denotes the indices of failed nodes among the systematic nodes of the $j$-th column, i.e., 
$$\mathcal{A}_\text{fail} = \{ t |  \Pi(t, j) \text{ is a failed systematic processor} \}$$
and $\mathcal{A}_\text{succ}$ denotes the indices of successful nodes among the systematic nodes of the $j$-th column. Similarly, we partition  $\{0, 1, ..., f-1\}= \mathcal{C}_\text{fail} \cup \mathcal{C}_\text{succ}$ to denote the failed and successful node indices within the checksum nodes on the $j$-th column. Note that $|\mathcal{A}_\text{fail}| = f_1$ and  $|\mathcal{C}_\text{fail} | = f_2$. 

First, \eqref{checksum_overhead1} can be rewritten as: 
\begin{align}
\label{recovery_eq1}
    C_{i,j}  - \sum_{t \in \mathcal{A}_{succ} } \tilg_{i,t} A_{t, j} &= \sum_{t \in \mathcal{A}_\text{fail} } \tilg_{i,t} A_{t, j} \forall i\in \mathcal{C}_\text{succ}
\end{align}
Let $\widehat{G} = [\tilg_{i, t}]_{i \in \mathcal{C}_\text{succ}, t \in \mathcal{A}_\text{fail} }$ bethe $f_1 \times f_1$ matrix. To decode $A_{t, j}$ for $t \in \mathcal{A}_\text{fail} $, we need to invert $\widehat{G}$. Note that $\widehat{G}$ is always invertible as any square submatrix of $\tilG$ is invertible. 
Denoting the $(i,j)$-th entry of $\widehat{G}^{-1}$ as $\widehat{g}^{-1}_{i,j}$, $A_{t, j}$ ($t \in \mathcal{A}_\text{fail}$) can be recovered from: 
\begin{align}
\nonumber 
    A_{t,j} &= \sum_{i \in \mathcal{C}_\text{succ}} \widehat{g}^{-1}_{t,i} (C_{i,j}  - \sum_{l \in \mathcal{A}_{succ} } \tilg_{t,l} A_{l, j} ) \\
    &= \sum_{i \in \mathcal{C}_\text{succ}} \widehat{g}^{-1}_{t,i} \cdot C_{i,j} - \sum_{l \in \mathcal{A}_{succ}}  \left(\sum_{i \in \mathcal{C}_\text{succ}} \widehat{g}^{-1}_{t,i} \cdot \tilg_{t,l} \right) \cdot A_{l, j}, \label{eq:A_recov}
\end{align}
which is a linear combination of $C_{i,j}$ ($i \in \mathcal{C}_\text{succ}$) and $A_{l, j}$ ($l \in \mathcal{A}_{succ}$). Computing the coefficients in \eqref{eq:A_recov} requires inverting $\widehat{G}$, which costs $\gamma \frac{2}{3} f_1^3$ and computing $\sum_{i \in \mathcal{C}_\text{succ}} \widehat{g}^{-1}_{t,i} \cdot \tilg_{t,l}$, which costs $\gamma f_1^2$. Once we have the coefficients, we can perform $f_1$ iterations of linear-combination reduce. The total cost of recovering the systematic data blocks is: 
\begin{equation}
    f_1 T_\text{lin-comb}(p_r, \frac{n^2}{P}) + \gamma(\frac{2}{3} f_1^3 + f_1^2) = f_1 \cdot T_\text{reduce}(p_r+1, \frac{n^2}{P})  + \gamma ( \frac{2}{3} f_1^3 + f_1^2 + \frac{f_1n^2}{P}).
\end{equation}

It is left to recover $f_2$ failed checksum blocks, so that the QR decomposition  is still fault-tolerant in the later iterations. This computation is the same as encoding the checksums (with the exception that we only have to compute a subset of checksums now). 
We thus defer readers to the analysis encoding overhead for more details, and only state the total cost of  this step, which is 
This requires: $f_2 T_\text{lin-comb}(p_r, \frac{n^2}{P}) = f_2 \cdot  T_{reduce}(p_r+1, \frac{n^2}{P}) + \gamma \frac{f_2 n^2}{P} $ .

%For recovering $f_2$ failed checksum blocks, we can now perform reduce operations from $p_r$ systematic nodes (including the recovered ones) to the  recovery checksum node. This is equivalent to $f_2$ operations of MPI\_reduce($p_r+1$, $\frac{n^2}{P}$). 

In total, the recovery overhead is thus: 
\begin{align*}
    T_\text{recov} =  f \cdot T_{reduce}(p_r+1, \frac{n^2}{P}) + \gamma (f_1^2  + \frac{2}{3} f_1^3+ \frac{fn^2}{P} ). 
\end{align*}
\end{proof}

\subsubsection{Proof of the main Theorem}

\begin{proof}
We have  the following bound of $T_{QR}$ from \eqref{QR_bound1}: 
\begin{align}
% \label{QR_bound1}
T_{QR} \geq  2\alpha n \log p_r +\beta \frac{1}{2}\frac{p_cn(n+1)}{P} + \gamma \frac{n^2(n+1)}{P}.
\end{align}
From Lemma \ref{overhead0}, we have:
\begin{align}
\nonumber
    T_\text{enc} &= f \cdot T_\text{reduce}(p_r+1, \frac{n^2}{P})  + \gamma \frac{fn^2}{P} + \gamma f (f+1) (p_r-f)\\
    &\leq 2 \alpha f \log (p_r+1) + 2 \beta \frac{fn(n+1)}{P} +\gamma  \frac{2 f n^2}{P} + \gamma f (f+1) (p_r-f) \label{enc_cost0}\\
    \nonumber
    &= (\frac{4f}{p_c} \cdot 2\alpha n \log p_r ) \frac{p_c}{n} \frac{\log(p_r+1)}{\log(p_r^4)} +(\frac{4f}{p_c} \cdot \beta \frac{1}{2}\frac{p_cn(n+1)}{P} ) \\
    \nonumber
    &\quad + (\frac{4f}{p_c} \cdot \gamma \frac{n^2(n+1)}{P}) \frac{p_c}{2(n+1)} +(\frac{f}{p_c} \cdot \gamma \frac{n^2(n+1)}{P}) \frac{(f+1) p_c (p_r-f)}{n^2 (n+1)}\\
    &\leq \frac{5 f}{p_c} T_{QR} \label{enc_cost1},
\end{align}
% From Theorem \ref{overhead0}, we have:
% \begin{align}
% \nonumber
%     T_\text{enc} &=  f \cdot T_\text{lin-comb}(p_r,\frac{n^2}{P}) + T_{allreduce} (p_r+f, f^2)+ \gamma f^2 \\
%     &\leq \big[ 2 \alpha f \log (p_r+1) + 2 \beta \frac{fn(n+1)}{P} +\gamma  \frac{2 f n^2}{P}  \big]+ \big[ 2\alpha \log(p_r+f) + 2\beta  f^2 +\gamma  f^2 \big]+ \gamma f^2\label{enc_cost0}\\
%     \nonumber
%     &= \bigg[ (\frac{4f}{p_c} \cdot 2\alpha n \log p_r ) \frac{p_c}{n} \frac{\log(p_r+1)}{\log(p_r^4)} +(\frac{4f}{p_c} \cdot \beta \frac{1}{2}\frac{p_cn(n+1)}{P} )+ (\frac{4f}{p_c} \cdot \gamma \frac{n^2(n+1)}{P}) \frac{p_c}{2(n+1)} \bigg] \\
%     \nonumber
%     &\quad + \bigg[ ( \frac{2f}{p_c} \cdot 2\alpha n \log p_r)  \frac{p_c}{n} \frac{\log(p_r+f)}{2 f \log(p_r)} + (\frac{2f}{p_c} \cdot \beta \frac{1}{2}\frac{p_cn(n+1)}{P} ) \frac{2f P}{n(n+1)} \\
%     \nonumber
%     &\quad + (\frac{2f}{p_c} \cdot \gamma \frac{n^2(n+1)}{P}) \frac{f p_c P}{n^2(n+1)} \bigg]\\
%     &\leq \frac{6f}{p_c} T_{QR} \label{enc_cost1},
% \end{align}
where for the last inequality we use $f \leq \frac{1}{2}\min\{p_r, p_c\}$ and $P = p_r p_c\leq n$.
% \begin{align}
% \nonumber
%     T_\text{enc} &= f \cdot T_\text{reduce}(p_r+1, \frac{n^2}{P})  + \gamma \big[ \frac{fn^2}{P} + f(p_r-f)(f+1) \big] \\
%     &\leq 2 \alpha f \log (p_r+1) + 2 \beta \frac{fn(n+1)}{P} +\gamma  \frac{2 f n^2}{P}  +  \gamma f(p_r-f)(f+1)  \label{enc_cost0}\\
%     \nonumber
%     &= (\frac{4f}{p_c} \cdot 2\alpha n \log p_r ) \frac{p_c}{n} \frac{\log(p_r+1)}{\log(p_r^4)} +(\frac{4f}{p_c} \cdot \beta \frac{1}{2}\frac{p_cn(n+1)}{P} )+ (\frac{4f}{p_c} \cdot \gamma \frac{n^2(n+1)}{P}) \frac{p_c}{2(n+1)} \\
%     &\quad + ( \frac{f}{p_c} \cdot \gamma \frac{n^2(n+1)}{P}) \frac{(p_r-f)p_c (f+1) P }{n^2 (n+1)}\\
%     &\leq \frac{4f}{p_c} T_{QR} \label{enc_cost1} \text{ (using $p_c \leq n$)}.
% \end{align}
From Lemma \ref{overhead0}, we have:
\begin{align}
\nonumber 
T_\text{post} &\leq f \cdot T_\text{broadcast}(p_r, \frac{2n^2}{P}) + f \cdot T_\text{reduce}(p_r, \frac{2n^2}{P}) +\gamma  \frac{4fn^2}{P} \\
\nonumber  
&\leq f (\sqrt{\alpha \log (p_r)} + \sqrt{\beta \frac{2n^2}{P}})^2  + 2 \alpha f \log (p_r) + 2 \beta \cdot\frac{2fn^2}{P} + \gamma  \frac{6fn^2}{P}\\ 
&\leq 4 \alpha f \log (p_r) + 8 \beta \cdot\frac{fn^2}{P}+ \gamma  \frac{6fn^2}{P}  \label{post_cost1}\\ 
\nonumber  
&\leq ( \frac{16f}{p_c} \cdot 2 \alpha  n \log (p_r)  ) \cdot \frac{p_c}{8n} +  \frac{16f}{p_c} \cdot \beta \frac{1}{2}\frac{p_cn(n+1)}{P} \\
\nonumber
&\quad    +(\frac{16f}{p_c} \cdot \gamma \frac{n^2(n+1)}{P}) \frac{3p_c}{8(n+1)})\\
&\leq  \frac{16f}{p_c} T_{QR} \label{post_cost2}   \text{ (using $p_c \leq n$)}.
\end{align}
Using \eqref{no_latency1}, we have:
\begin{align}
    T^\alpha_\text{comp} &= O(1) \label{comp_cost0}.
\end{align}
Combining the bounds \eqref{enc_cost1},  \eqref{post_cost2} and \eqref{comp_cost1}, we obtain that:
\begin{align}
\nonumber 
T_\text{coding} &= T_\text{enc}+ T_\text{post} +T_\text{comp} \leq (\frac{6f}{p_c}+\frac{16f}{p_c}  + \frac{f}{p_r})T_{QR} = (\frac{22f}{p_c} + \frac{f}{p_r})T_{QR}, \\
\nonumber 
\therefore T_\text{coding} &=  O\left(\frac{f}{p_r} + \frac{f}{p_c}\right) T_\text{QR}.
\end{align}
From Lemma \ref{overhead0}, we have:
\begin{align}
\nonumber 
T_\text{recov} &\leq  T_\text{reduce}(p_r, \frac{fn^2}{P})  + \gamma  [\frac{fn^2}{P}  + \frac{2}{3} \frac{fn^2}{P} + \frac{fn^2}{P} ] \\
\nonumber 
&\text{ (using $f_1^2 \leq f_1^3 \leq f^3 \leq f p_r p_c =fP\leq \frac{fn^2}{P}$)}\\
\nonumber 
&\leq 3 T_\text{enc} \leq \frac{12 f}{p_c} T_\text{QR}, \\
\nonumber 
\therefore T_\text{recov}&= O\left(  \frac{ f}{p_c}  \right) T_\text{QR}
\end{align}  
From \eqref{enc_cost0}, \eqref{post_cost1} and \eqref{comp_cost0} , we have:
\begin{align*}
    T^\alpha_\text{enc} &\leq 2\alpha f\log(p_r+1) + 2\alpha \log(p_r+f) = O(\alpha f\log(p_r)),\\
    T^\alpha_\text{post} &\leq  4\alpha f \log(p_r) = O(\alpha f\log(p_r)),\\
    % T^\alpha_\text{comp} &= O(1) \label{comp_cost0}\\ 
    T^\alpha_\text{coding} &=  T^\alpha_\text{enc} + T^\alpha_\text{post} +T^\alpha_\text{comp} 
    = O(\alpha f\log(p_r)) \\
    &= O\left(\frac{f}{n}\right) 2\alpha n \log(p_r) \leq  O\left(\frac{f}{n}\right) T^\alpha_\text{QR} \text{ (by \eqref{QR_bound1})},\\
    T^\alpha_\text{coding} &=  O\left(\frac{f}{n}\right) T^\alpha_\text{QR}. 
\end{align*}
\end{proof}